\documentclass[preprint,12pt]{elsarticle}

\usepackage{amssymb}
\usepackage{gastex}
\usepackage{amsmath}
\usepackage{upref}
\usepackage{multirow}
\usepackage{multicol}
\usepackage{url}
\usepackage[all]{xy}
\usepackage{vaucanson-g}
\usepackage{color}
\usepackage[utf8]{inputenc}
\usepackage{makeidx}
\usepackage{theorem}

\newtheorem{theorem}{Theorem}

\newtheorem{proposition}[theorem]{Proposition}
\newtheorem{corollary}[theorem]{Corollary}
{\theorembodyfont{\rmfamily}%
  \newtheorem{example}[theorem]{Example}
   }
\newenvironment{proof}{\noindent\textit{Proof.}}
{\QED\vskip\theorempostskipamount} 
\newenvironment{proofof}[1]{\noindent\textit{Proof
    \protect{#1}.}}
                       {\QED\vskip\theorempostskipamount}
\def\petitcarre{\vrule height4pt width 4pt depth0pt}
\def\QED{\relax\ifmmode\eqno{\hbox{\petitcarre}}\else{%
  \unskip\nobreak\hfil\penalty50\hskip2em\hbox{}\nobreak\hfil
  \petitcarre
  \parfillskip=0pt \finalhyphendemerits=0\par\smallskip}
  \fi}

\newcommand\A{\mathcal{A}}

\newcommand\E{\mathcal{E}}
\newcommand\T{\mathcal{T}}
\newcommand\C{\mathcal{C}}

\newcommand\G{\mathcal{G}}

\newcommand\LL{\mathcal{L}}
\newcommand\RR{\mathcal{R}}
\newcommand\CR{\mathcal{CR}}
\newcommand\MR{\mathcal{MR}}

\newcommand{\Z}{\mathbb{Z}}
\newcommand{\R}{\mathbb{R}}

\def\un(#1){\underline{#1}\,}
\DeclareMathOperator{\Card}{Card}

\definecolor{ivoire}{rgb}{0.99,0.99,0.8}

\def\id{{\rm id}}

\definecolor{light-gray}{gray}{0.7}
\definecolor{forestgreen}{RGB}{34,139,34}
\definecolor{golden}{rgb}{1.0, 0.84, 0.0}

\numberwithin{theorem}{section}
\numberwithin{equation}{section}
\numberwithin{figure}{section}
\numberwithin{table}{section}

\begin{document}
\begin{frontmatter}

\title{Specular sets}

\author{Val\'erie Berth\'e$^1$, Clelia De Felice$^2$, Vincent Delecroix$^3$, \\ Francesco Dolce$^4$, Julien Leroy$^5$, Dominique Perrin$^4$, \\ Christophe Reutenauer$^6$, Giuseppina Rindone$^4$ \\ $\,$\\ $^1$ CNRS, Universit\'e Paris Diderot \\ $^2$ Universit\`a degli Studi di Salerno \\ $^3$ CNRS, Universit\'e de Bordeaux \\ $^4$ Universit\'e Paris-Est \\ $^5$ Universit\'e de Li\`ege \\ $^6$ Universit\'e du Qu\'ebec \`a Montr\'eal}

\begin{abstract}
We introduce specular sets.
These are subsets of groups which form a natural generalization of free groups.
These sets of words are an abstract generalization of the natural codings of interval exchanges and of linear involutions.
We consider two important families of sets contained in specular sets: sets of return words and maximal bifix codes.
For both families we prove several cardinality results as well as results concerning the subgroup generated by these sets.
\end{abstract}

\begin{keyword}
Tree sets; Return words; Bifix codes; Linear involutions; Specular groups; Free group.
\end{keyword}

\end{frontmatter}

\tableofcontents

\section{Introduction}

We have studied in a series of papers initiated in~\cite{BerstelDeFelicePerrinReutenauerRindone2012} the links between uniformly recurrent languages, subgroups of free groups and bifix codes.
In this paper, we continue this investigation in a situation which involves groups which are not free anymore.
These groups, named here specular, are free products of a free group and of a finite number of cyclic groups of order two.
These groups are close to free groups and, in particular, the notion of a basis in such groups is clearly defined.
It follows from the Kurosh subgroup theorem that any subgroup of a specular group is specular.
A specular set is a subset of such a group which generalizes the natural codings of linear involutions studied in~\cite{BertheDelecroixDolcePerrinReutenauerRindone2014}.
It is a set of words stable by taking the inverse and defined in terms of restrictions on the extensions of its elements.

The main results of this paper are Theorems~\ref{theoremReturns} and~\ref{theoremFiniteIndex}, referred to as the First Return Theorem and the Finite Index Basis Theorem.
The first one asserts that the set of return words to a given word in a recurrent specular set is a basis of a subgroup of index 2, called the even subgroup.
The last one characterizes the symmetric bases of subgroups of finite index of specular groups contained in a specular set $S$ as the finite $S$-maximal symmetric bifix codes contained in $S$. This generalizes the analogous result proved initially in~\cite{BerstelDeFelicePerrinReutenauerRindone2012} for Sturmian sets and extended in~\cite{BertheDeFeliceDolceLeroyPerrinReutenauerRindone2013b} to a more general class of sets, containing both Sturmian sets and interval exchange sets.

There are two interesting features of the subject of this paper.

In the first place, some of the statements concerning the natural codings of linear involutions can be proved using geometric methods, as shown in a separate paper~\cite{BertheDelecroixDolcePerrinReutenauerRindone2014}.
This provides an interesting interpretation of the groups playing a role in the natural codings (these groups are generated either by return words or by maximal bifix codes) as fundamental groups of some surfaces.
The methods used here are, however, purely combinatorial.

In the second place, the abstract notion of specular set gives rise to groups called here specular.
These groups are natural generalizations of free groups, and are free products of a finite number of copies of $\Z$ and of $\Z/2\Z$.
They are called \emph{free-like} in~\cite{Bartholdi2014} and appear at several places in~\cite{Harpe2000}.

The idea of considering recurrent sets of reduced words invariant by taking inverses is connected with the notion of $G$-full words of~\cite{PelantovaStarosta2014} (see Section~\ref{sec:palindromes}).

The paper is organized as follows.
In Section~\ref{sectionPreliminaries}, we recall some notions concerning words, extension graphs and bifix codes.
We define the notion of characteristic which is the Euler characteristic of the extension graph of the empty word.
We consider tree sets of characteristic $1$ or $2$ (tree sets of characteristic $1$ are introduced in~\cite{BertheDeFeliceDolceLeroyPerrinReutenauerRindone2013a}, while the case of arbitrary characteristic is treated in~\cite{DolcePerrin2016}).

In Section~\ref{sectionSpecularGroups}, we introduce specular groups, which form a family with properties very close to free groups.
We deduce from the Kurosh subgroup theorem that any subgroup of a specular group is specular (Theorem~\ref{theoremKurosh}).

In Section~\ref{sectionSpecular} we introduce specular sets.
We recall several results from~\cite{DolcePerrin2015} and~\cite{DolcePerrin2016} concerning the cardinality of some sets included in neutral sets, namely bifix codes (Theorems~\ref{theoremCardinality} and~\ref{corollaryCardinality}).
We give a construction which allows to build specular sets from a tree set of characteristic $1$ using a transducer called doubling transducer (Theorem~\ref{theoremDoubling}).
We make a connection with the notion of $G$-full words introduced in~\cite{PelantovaStarosta2014} and related to the palindromic complexity of~\cite{DroubayJustinPirillo2001}.

In Section~\ref{sectionInvolutions} we recall the definition of a linear involution introduced in~\cite{DanthonyNogueira1988} and we show that the natural coding of a linear involution without connections is a specular set (Theorem~\ref{theoremInvolutionSpecular}).

In Section~\ref{sec:return} we introduce three variants of the notion of set of return words.
We prove several cardinality results concerning these sets (Theorems~\ref{theoremCardRightReturns}, \ref{theoremCardReturns}, \ref{theoremCardFirstMixed}).
We prove that the set of return words to a given word forms a basis of the even subgroup (Theorem~\ref{theoremReturns} referred to as the First Return Theorem) and that the mixed return words form a monoidal basis of the specular group (Theorem \ref{theoremFirstMixed}).

In Section~\ref{sec:groups} we prove several results concerning subgroups generated by bifix codes.
We prove that a set closed by taking inverses is acyclic if and only if any symmetric bifix code is free (Theorem~\ref{theoremFreeness}).
Moreover, we prove that in such a set, for any finite symmetric bifix code $X$, the free monoid $X^*$ and the free subgroup $\langle X \rangle$ have the same intersection with $S$ (Theorem~\ref{theoremSaturation}).

Finally, in Section~\ref{sectionFiniteIndex}, we prove the Finite Index Basis Theorem (Theorem~\ref{theoremFiniteIndex}) and a converse (Theorem~\ref{propositionConverseFIB}).

\paragraph{Acknowledgments}
This paper is an extended version of a conference paper~\cite{words2015}.
The authors thank  Laurent Bartholdi and Pierre de la Harpe for useful indications.
This work was supported by grants from R\'egion \^{I}le-de-France, the ANR projects Dyna3S and Eqinocs.

\section{Preliminaries}
\label{sectionPreliminaries}
In this section, we first recall some notions on sets of words including recurrent, uniformly recurrent and tree sets.
We also recall some definitions and properties concerning bifix codes.

\subsection{Extension graphs}
Let $A$ be a finite alphabet.
We denote by $A^*$ the free monoid on $A$.
We denote by $\varepsilon$ the empty word.
The \emph{reversal} of a word $w=a_1a_2\cdots a_n$ with $a_i\in A$ is the word $\tilde{w} = a_n\cdots a_2a_1$.
A word $w$ is said to be a \emph{palindrome} if $w = \tilde{w}$.

A set of words on the alphabet $A$ is said to be \emph{factorial} if it contains the alphabet $A$ and all the factors of its elements. 

An \emph{internal factor} of a word $x$ is a word $v$ such that $x=uvw$ with $u,w$ nonempty.

Let $S$ be a  set of words on the alphabet $A$.
For $w\in S$, we denote
\begin{eqnarray*}
L_S(w) &=& \{a\in A\mid aw\in S\}\\
R_S(w) &=& \{a\in A\mid wa\in S\}\\
B_S(w) &=& \{(a,b)\in A\times A \mid awb\in S\}
\end{eqnarray*}
and further
\begin{displaymath}
\ell_S(w)=\Card(L_S(w)),\quad r_S(w)=\Card(R_S(w)),\quad b_S(w)=\Card(B_S(w)).
\end{displaymath}
We omit the subscript $S$ when it is clear from the context.
A word $w$ is \emph{right-extendable} if $r(w)>0$,
\emph{left-extendable} if $\ell(w)>0$ and \emph{biextendable} if
$b(w)>0$. A 
factorial set
$S$ is called \emph{right-extendable}
(resp. \emph{left-extendable}, resp. \emph{biextendable}) if every word in $S$ is
right-extendable (resp. left-extendable, resp. biextendable).

A word $w$ is called \emph{right-special} if $r(w)\ge 2$. 
It is called \emph{left-special} if $\ell(w)\ge 2$. 
It is called \emph{bispecial} if it is
both left-special and right-special.

For $w\in S$, we denote 
\begin{displaymath}
m_S(w) = b_S(w)-\ell_S(w)-r_S(w)+1.
\end{displaymath}
The word $w$ is called \emph{weak} if $m_S(w)<0$, \emph{neutral} if $m_S(w)=0$ and \emph{strong} if $m_S(w)>0$.

We say that a factorial set $S$ is \emph{neutral} if every nonempty word in $S$ is neutral.
The \emph{characteristic} of $S$ is the integer $\chi(S)=1-m_S(\varepsilon)$.
Thus a neutral set of characteristic $1$ is such that all words (including the empty word) are neutral. This what is called a neutral
set in~\cite{BertheDeFeliceDolceLeroyPerrinReutenauerRindone2013a}.

A set of words $S\ne\{\varepsilon\}$ is \emph{recurrent} if it is factorial and if for any $u,w \in S$, there is a $v \in S$ such that $uvw \in S$.
An infinite factorial set is said to be \emph{uniformly recurrent} if for any word $u\in S$ there is an integer $n \geq 1$ such that $u$ is a factor of any word of $S$ of length $n$.
A uniformly recurrent set is recurrent.

In~\cite{DolcePerrin2016} it is proved that the converse is true for neutral sets.
As all sets we will deal with are neutral, we usually omit the term ``uniformly'' and just mention whenever we suppose them to be recurrent.

\begin{theorem}[\cite{DolcePerrin2016}]
\label{theo:recurrentur}
A recurrent neutral set is uniformly recurrent
\end{theorem}

The \emph{factor complexity} of a factorial set $S$ of words on an alphabet $A$ is the sequence $p_n=\Card(S\cap A^n)$.
Let $s_n=p_{n+1}-p_n$ and $t_n=s_{n+1}-s_n$ be respectively the first and second order differences sequences of the sequence $p_n$. 

The following result is from~\cite{Cassaigne1997} (see also \cite{BertheRigo2010}, Theorem 4.5.4).

\begin{proposition}
\label{propCANT}
Let $S$ be a factorial set on the alphabet $A$.
One has $t_n = \sum_{w\in S\cap A^n}m(w)$ and $s_n = \sum_{w\in S\cap A^n}(r(w)-1)$ for all $n\ge 0$.
\end{proposition}

Let $S$ be a biextendable set of words.
For $w\in S$, we consider define the undirected graph $\E_S(w)$, or simply $\E(w)$ when $S$ is clear from the context, having as set of vertices the disjoint union of $L(w)$ and $R(w)$ and edges the pairs $(a,b)\in B(w)$.
This graph is called the \emph{extension graph} of $w$.
We sometimes denote $1\otimes L(w)$ and $R(w)\otimes 1$ the copies of $L(w)$ and $R(w)$ used to define the set of vertices of $\E(w)$.
We note that, since $\E(w)$ has $\ell(w)+r(w)$ vertices and $e(w)$ edges, the number $1-m(w)$ is the Euler characteristic of the graph $\E(w)$\footnote{We consider here graphs as 1-dimensional complexes and thus they have no faces.}.

If the extension graph $\E(w)$ is acyclic, then $m(w)\le 0$.
Thus $w$ is weak or neutral.
More precisely, one has in this case that $c=1-m(w)$ is the number of connected components of the graph $\E(w)$.

A biextendable set $S$ is called \emph{acyclic} if for every $w\in S$, the graph $\E(w)$ is acyclic.

A biextendable set $S$ is called a \emph{tree set} of characteristic $c$ if for any nonempty $w \in S$, the graph $\E(w)$ is a tree and if $\E(\varepsilon)$ is a union of $c$ trees (the definition of tree set in~\cite{BertheDeFeliceDolceLeroyPerrinReutenauerRindone2013a} corresponds to a tree set of characteristic $1$).
Note that a tree set of characteristic $c$ is a neutral set of characteristic $c$. We focus here on characteristic $1$ or $2$ (specular sets, that we will introduce in Section~\ref{sectionSpecular}, are tree sets of characteristic $2$ with some symmetric properties).

An infinite word is \emph{episturmian} if the set of its factors is closed under reversal and if it contains for each $n$ at most one word of length $n$ which is right-special.
It is a \emph{strict episturmian} word if the set of its factors has exactly one right-special word of each length and moreover each of these words $u$ is such that $r(u)=\Card(A)$ (see~\cite{BerstelDeFelicePerrinReutenauerRindone2012}).

A \emph{Sturmian set} is the set of factors of a strict episturmian word.
Any Sturmian set is a recurrent tree set of characteristic $1$ (see~\cite{BertheDeFeliceDolceLeroyPerrinReutenauerRindone2013a}). 

\begin{example}
\label{exampleFibo}
Let $A=\{a,b\}$. The \emph{Fibonacci morphism} is the morphism $f:A^*\rightarrow A^*$ defined by $f(a)=ab$ and $f(b)=a$.
The \emph{Fibonacci word} is the fixed-point $f^\omega(a)$ of the Fibonacci morphism.
Its set of factors is a Sturmian set (see~\cite{Lothaire2002}).
\end{example}

\subsection{Bifix codes}
\label{sectionBifix}

A \emph{prefix code} is a set of nonempty words which does not contain any
proper prefix of its elements. A \emph{suffix code} is defined symmetrically.
A \emph{bifix code} is a set which is both a prefix code and a suffix code
(see~\cite{BerstelPerrinReutenauer2009} for a more detailed introduction).

A \emph{coding morphism} for a prefix code $X$ on the alphabet $A$
is a morphism $f:B^*\rightarrow A^*$ which maps bijectively $B$ onto $X$.

Let $S$ be a recurrent set.
A prefix (resp. bifix) code $X\subset S$ is $S$-maximal if it is not properly contained
in a prefix (resp. bifix) code $Y\subset S$. 
Since $S$ is recurrent, a finite $S$-maximal bifix code is also an
$S$-maximal prefix code 
(see~\cite{BerstelDeFelicePerrinReutenauerRindone2012}, Theorem 4.2.2).

For example, for any $n\ge 1$, the set $X=S\cap A^n$ is an $S$-maximal
bifix code.

Let $X$ be a bifix code. Let $Q$ be the set of words without any suffix in $X$
and let $P$ be the set of words without any prefix in $X$.
A \emph{parse} of a word $w$ with respect to a bifix code $X$
is a triple $(q,x,p)\in Q\times X^*\times P$ such that $w=qxp$. 
We denote by $d_X(w)$ the number of parses of a word $w$
with respect to $X$. The $S$-degree of $X$, denoted $d_X(S)$, is
the maximal number of parses with respect to $X$ of a word of $S$.

For example, the set $X=S\cap A^n$
has $S$-degree $n$.

Let $S$ be a recurrent set and let $X$ be a finite bifix code. By Theorem 4.2.8
in~\cite{BerstelDeFelicePerrinReutenauerRindone2012}, $X$ is $S$-maximal
if and only if its $S$-degree is finite. Moreover, in this case,
a word $w\in S$ is such that $d_X(w)<d_X(S)$ if and only if it is an internal factor of a word of $X$. 

The \emph{kernel} of a bifix code $X$ is the set of words of $X$ which
are internal factors of $X$.

We will use bifix codes in relation with
a more general version of extension graphs (see
\cite{BertheDeFeliceDolceLeroyPerrinReutenauerRindone2013a}). For two sets of words $X,Y$ and a word $w \in S$, we denote $L_S^X(w) = \{x \in X \mid xw \in S \}, \ R_S^Y(w) = \{y \in Y \mid wy \in S\}, \ B_S^{X,Y}(w) = \{ (x,y) \in X \times Y \mid xwy \in S \}$.
We also define $\E_S^{X,Y}(w)$
as the undirected graph on the set of vertices which is the disjoint union of $L_S^X(w)$ and $R_S^{Y}(w)$ and edges in $B_S^{X,Y}(w)$.
Set further
$$
 \ell_S^X(w) = \Card(L_S^X(w)),\ r_S^Y(w) = \Card(R_S^Y(w)),\
b_S^{X,Y}(w) = \Card(B_S^{X,Y}(w)).
$$
Finally, for a word $w$, we denote $m_S^{X,Y}(w) = b_S^{X,Y}(w) - \ell_S^X(w) - r_S^Y(w)+1$.
Note that $\E_S^{A,A}(w) = \E_S(w)$, $m_S^{A,A}(w) = m_S(w)$, and so on.

We will use below the following result.

\begin{proposition}
\label{propositionExtended}
Let $S$ be a recurrent set, let $X\subset S$ be a finite $S$-maximal suffix code and let $Y\subset S$ be a finite $S$-maximal prefix code.
\begin{enumerate}
\item If $\E_S(x)$ is acyclic, then  $\E_S^{X,Y}(x)$ is acyclic.
\item If $S$ is neutral, then  $m_S^{X,Y}(w) = m_S(w)$ for every $w\in S$.
\end{enumerate}
\end{proposition}
\begin{proof}
Statement 1 follows from Proposition 3.7 in~\cite{BertheDeFeliceDolceLeroyPerrinReutenauerRindone2013a}. 
Statement 2 is Proposition 6.2 in~\cite{DolcePerrin2016}.
\end{proof}

Observe that the condition that $X$ (resp. $Y$) is an $S$-maximal suffix (resp. prefix) code is only necessary for Assertion 2
(for Assertion 1, $X$ (resp. $Y$) may be an arbitrary suffix (resp. prefix)
code). Observe also
 that this condition can be replaced by the condition that $X$ (resp. $Y$) is an
$Sw^{-1}$-maximal suffix code (resp. a $w^{-1}S$-maximal prefix code),
where $Sw^{-1}=\{u\in S\mid uw\in S\}$ and symmetrically $w^{-1}S=\{u\in S\mid
wu\in S\}$.

\section{Specular groups}
\label{sectionSpecularGroups}
In this section, we introduce specular groups and we prove some properties of this family of groups.
In particular, using the Kurosh subgroup theorem, we prove that any subgroup of a specular group is specular (Theorem~\ref{theoremKurosh}).

\subsection{Definitions}
\label{subsec:speculardefinitions}
We consider an alphabet $A$ with an involution $\theta:A\rightarrow A$, possibly with some fixed points.
We also consider the group $G_\theta$ generated by $A$ with the relations $a\theta(a)=\varepsilon$ for every $a\in A$.
Thus $\theta(a)=a^{-1}$ for $a\in A$.
The set $A$ is called a \emph{natural} set of generators of $G_\theta$. 

When $\theta$ has no fixed point, we can set $A=B\cup B^{-1}$
by choosing a set of representatives of the orbits of $\theta$
for the set $B$. The group $G_\theta$ is then the free
group on $B$, denoted $F_B$. In general, the group $G_\theta$ is a free product of
a free group and a finite number of copies of $\Z/2\Z$, that
is $G_\theta=\Z^{*i}*(\Z/2\Z)^{*j}$ where $i$  is the number
of orbits of $\theta$ with two elements  and $j$ the number of its fixed points. 
Such a group will be called a \emph{specular group} of type $(i,j)$.
These groups are very close to free groups, as we will see.
The integer $\Card(A)=2i+j$ is called the \emph{symmetric rank} of the specular group $\Z^{*i}*(\Z/2\Z)^{*j}$.

\begin{proposition}
\label{pro:isospecular}
Two specular groups are isomorphic if and only if they have the same type.
\end{proposition}
\begin{proof}
The commutative image of a group of type $(i,j)$ is $\Z^i \times (\Z/2\Z)^j$ and the uniqueness of $i,j$ follows from the fundamental theorem of finitely generated Abelian groups. 
\end{proof}

\begin{example}
\label{exampleSpecularGroup}
Let $A=\{a,b,c,d\}$ and let $\theta$ be the involution which exchanges $b,d$
and fixes $a,c$. Then $G_\theta=\Z*(\Z/2\Z)^2$ is a specular group of symmetric
rank
$4$.
\end{example}
The Cayley graph of a specular group $G_\theta$ with respect to the set
of natural generators $A$ is a regular tree where each vertex has degree
$\Card(A)$.
The specular groups are actually
characterized by this property (see~\cite{Harpe2000}).

\subsection{Subgroups}
\label{subsec:sub}
By the Kurosh subgroup theorem, any subgroup of a free product $G_1*G_2*\cdots *G_n$ is itself a free product of a free group and of groups conjugate to subgroups of the $G_i$ (see~\cite{MagnusKarrassSolitar2004}).
Thus, we have, replacing  the Nielsen-Schreier Theorem of free groups, the following result.

\begin{theorem}
\label{theoremKurosh}
Any subgroup of a specular group is specular.
\end{theorem}

It also follows from the Kurosh subgroup theorem that the elements of order $2$ in a specular group $G_\theta$ are the conjugates of the $j$ fixed points of $\theta$ and this number is thus the number
of conjugacy classes of elements of order $2$.
Indeed, an element of order $2$ generates a subgroup conjugate to one of the subgroups generated by the letters of order 2.

Any specular group $G=G_\theta$ has a free subgroup of index $2$. Indeed, let
$H$ be the subgroup formed of the reduced words of even length. It has
clearly index $2$. It is free because it does not contain any element
of order $2$ (such an element is conjugate to a fixed point of $\theta$
and thus is of odd length). 

A group having a free subgroup of finite index is called \emph{virtually
free} (see \cite{Harpe2000}). On the other hand, a finitely generated group is
said to be \emph{context-free} if, for some presentation, the set
of words equivalent to $\varepsilon$ is a context-free language.
By Muller and Schupp's theorem,
a finitely generated group is virtually free if and only if
it is context-free \cite{MullerSchupp1983}. 
Thus a specular group is context-free. One may
verify this directly as follows. A
context-free grammar generating the words equivalent to $\varepsilon$ for the
natural presentation of a specular group $G=G_\theta$ is the
grammar with one nonterminal symbol $\sigma$ and the rules
\begin{displaymath}
\sigma\rightarrow a\sigma a^{-1}\sigma\quad (a\in A),\quad \sigma\rightarrow \varepsilon.
\end{displaymath}
The proof that this grammar generates the set of words
equivalent to $\varepsilon$ is similar to that used in~\cite{Berstel1979}
for the so-called Dyck-like languages.

We will need two more properties of specular groups.
Both are well-known to hold for free groups.

A group $G$ is called \emph{residually finite} if for every element $g\ne \varepsilon$ of $G$, there is a morphism $\varphi$ from $G$ onto a finite group such that $\varphi(g)\ne \varepsilon$.

\begin{proposition}\label{propResiduallyFinite}
Any specular group is residually finite.
\end{proposition}
\begin{proof}
Let $K$ be a free subgroup of index $2$ in the specular group $G$.
Let $g\ne 1$ be in $G$. If $g\notin K$, then the image of $g$ in $G/K$
is nontrivial. Assume $g\in K$. Since $K$ is free, it is residually finite.
Let $N$ be a normal subgroup of finite index
of $K$ such that $g\notin N$. Consider the representation of $G$
on the right cosets of $N$. Since $g\notin N$, the image of $g$
in this finite group is nontrivial.
\end{proof}

A group $G$ is said to be \emph{Hopfian} if any surjective morphism from $G$ onto $G$ is also injective.
By a result of Malcev, any finitely generated residually finite group is Hopfian (see~\cite{LyndonSchupp2001}, p. 197).
We thus deduce from Proposition~\ref{propResiduallyFinite} the following result.

\begin{proposition}
\label{pro:hopf}
A specular group is Hopfian.
\end{proposition}

\subsection{Monoidal basis}
\label{subsec:basis}
A word on the alphabet $A$ is $\theta$-\emph{reduced} (or simply \emph{reduced}) if it has no factor of the form $a\theta(a)$ for $a\in A$.
It is clear that any element of a specular
group is represented by a unique reduced word.

A subset of a group $G$ is called \emph{symmetric} if it is closed under taking inverses.
A set $X$ in a specular group $G$ is called a \emph{monoidal basis} of $G$ if it is symmetric, if the monoid that it generates is $G$  and if any product $x_1x_2\cdots x_m$ of elements of $X$ such that $x_kx_{k+1}\ne \varepsilon$ for $1\le k\le m-1$ is distinct of $\varepsilon$.

\begin{example}
The alphabet $A$ is a monoidal basis of $G_\theta$.
\end{example}

The previous example shows that the symmetric rank of a specular group is the cardinality of any monoidal basis (two monoidal bases have the same cardinality since the type is invariant by isomorphism by Proposition~\ref{pro:isospecular}).

Let $H$ be a subgroup of a specular group $G$.
Let $Q$ be a set of reduced words on $A$ which is a prefix-closed set of representatives of the right cosets $Hg$
of $H$.
Such a set is traditionally called a \emph{Schreier transversal} for $H$
(the proof of its existence is classical in the free group and it
is the same in any specular group).

Let
\begin{equation}
X=\{paq^{-1}\mid a\in A, p,q\in Q, pa\not\in Q, pa\in Hq\}.\label{setSchreier}
\end{equation}
Each word $x$ of $X$ has a unique factorization $paq^{-1}$ with $p,q\in Q$
and $a\in A$. The letter $a$ is called the \emph{central part} of $x$.
The set $X$
is a monoidal basis of $H$, called the \emph{Schreier basis}
relative to $Q$.

\begin{proposition}
Let $H$ and $Q$ be as above and let $X$ be a Schreier basis relative to $Q$.
Then $X$ is closed by taking inverses.
\end{proposition}
\begin{proof}
Let $x=paq^{-1} \in X$, then  $x^{-1}=qa^{-1}p^{-1}$.
We cannot have $qa^{-1}\in Q$ since otherwise $p\in Hqa^{-1}$ implies
$p=qa^{-1}$ by uniqueness of the coset representative and finally $pa\in Q$.
It generates $H$ as a monoid because if $x=a_1a_2\cdots a_m\in H$
with $a_i\in A$,
then $x=(a_1p_1^{-1})(p_1a_2p_2^{-1})\cdots(p_{m-1}a_m)$ with $a_1\cdots a_k\in Hp_k$
for $1\le k\le m-1$ is a factorization of $x$ in elements of $X\cup \{1\}$.
Finally, if a product $x_1x_2\cdots x_m$ of elements of $X$ 
is equal to $1$, then $x_kx_{k+1}=1$ for some index $k$
since the central part $a$ never
cancels in a product of two elements of $X$.
\end{proof}

One can deduce directly Theorem~\ref{theoremKurosh} from these properties of $X$.

\begin{proofof}{of Theorem~\ref{theoremKurosh}}
Let $H$ be a subgroup of a specular group $G$, $Q$ be a Shreier transversal for $H$ and $X$ be the Schreier basis relative to $Q$.
Let $\varphi:B\rightarrow X$ be a bijection from a set $B$ onto $X$ which extends to a morphism from $B^*$ onto $H$.
Let $\sigma:B\rightarrow B$ be the involution sending each $b$ to $c$ where $\varphi(c)=\varphi(b)^{-1}$.
Since the central parts never cancel, if a nonempty word $w\in B^*$  is $\sigma$-reduced then $\varphi(w)\ne \varepsilon$.
This shows that $H$ is isomorphic to the group $G_\sigma$.
Thus $H$ is specular.
\end{proofof}

If $H$ is a subgroup of index $n$ of a specular group $G$ of symmetric rank $r$,
the symmetric rank $s$ of $H$ is
\begin{equation}
s=n(r-2)+2. \label{SchreierSpecular}
\end{equation}
This formula replaces Schreier's Formula (which corresponds to the case $j=0$). It can be proved as follows.
Let $Q$ be a Schreier transversal for $H$ and let $X$ be the
corresponding Schreier basis.
 The number
of elements of $X$ is $nr-2(n-1)$. Indeed, this
is the number of pairs $(p,a)\in Q\times A$
 minus the $2(n-1)$ pairs $(p,a)$ such that $pa\in Q$ with $pa$
reduced or $pa\in Q$ with $pa$ not reduced.  This gives Formula
\eqref{SchreierSpecular}.

\begin{example}
\label{exampleEvenLength}
Let $G$ be the specular group of Example~\ref{exampleSpecularGroup}.
Let $H$ be the subgroup formed by the elements represented by a reduced word of even length.
The set $Q=\{\varepsilon,a\}$ is a prefix-closed set of representatives of the two cosets of $H$.
The representation of $G$ by permutations on the cosets of $H$ is represented in Figure~\ref{figureSchreier}.

\begin{figure}[hbt]
\centering\gasset{Nadjust=wh}
\begin{picture}(20,20)(0,-7)
\node(1)(0,0){$\varepsilon$}
\node(a)(20,0){$a$}

\drawedge[curvedepth=5](1,a){$a,b,c,d$}
\drawedge[curvedepth=5](a,1){$a,b,c,d$}
\end{picture}
\caption{The representation of $G$ by permutations on the cosets of $H$.}
\label{figureSchreier}
\end{figure}
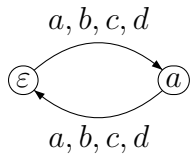
The monoidal basis corresponding to Formula \eqref{setSchreier} is 
$X=\{ab,ac,ad,ba,ca,da\}$.
The symmetric rank of $H$ is $6$, in agreement with Formula~\eqref{SchreierSpecular} and $H$ is a free group of rank $3$.
\end{example}

\begin{example}
Let again $G$ be the specular group of Example~\ref{exampleSpecularGroup}.
Consider now the subgroup $K$ stabilizing $1$ in the representation of $G$ by permutations on the set $\{1,2\}$ of Figure~\ref{figureSchreier2}.

\begin{figure}[hbt]
\centering\gasset{Nadjust=wh}
\begin{picture}(20,20)(0,-7)
\node(1)(0,0){$1$}\node(2)(20,0){$2$}

\drawloop[loopangle=180](1){$a,c$}
\drawedge[curvedepth=5](1,2){$b,d$}
\drawedge[curvedepth=5](2,1){$b,d$}
\drawloop[loopangle=0](2){$a,c$}
\end{picture}
\caption{The representation of $G$ by permutations on the cosets of $K$.}
\label{figureSchreier2}
\end{figure}
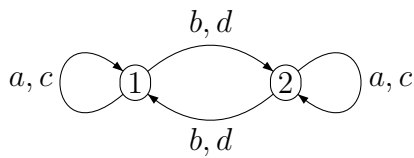

We choose $Q=\{\varepsilon,b\}$.
The set $X$ corresponding to Formula \eqref{setSchreier} is $X=\{a,bad,bb,bcd,c,dd\}$.
The group $K$ is isomorphic to $\Z*(\Z/2\Z)^{*4}$.
\end{example}

The following result, which will be used later (Section~\ref{sec:return}), is a consequence of Proposition~\ref{pro:hopf}.

\begin{proposition}\label{propSymBasis}
Let $G$ be a specular group of type $(i,j)$ and let $X\subset G$ be a symmetric set with $2i+j$ elements. If $X$
generates $G$, it is a monoidal basis of $G$.
\end{proposition}
\begin{proof}
Let $A$ be a set of natural generators of $G$.
Considering the commutative image of $G$, we obtain that $X$
contains $j$ elements of order $2$. Thus there is a bijection
$\varphi$ from $A$ onto $X$ such that $\varphi(a^{-1})=\varphi(a)^{-1}$
for every $a\in A$. The map $\varphi$ extends to a morphism
 from $G$ to $G$ which is surjective since $X$ generates $G$. Then $\varphi$ being surjective, it also injective
since $G$ is Hopfian, and thus $X$ is a monoidal basis of $G$.
\end{proof}

\section{Specular sets}
\label{sectionSpecular}
In this section, we introduce specular sets.
We introduce odd and even words and the even code which play an important part in the sequel.
We prove that the decoding of a recurrent specular set by the even code is a union of two recurrent tree sets of characteristic $1$ (Theorem~\ref{theoremDecodingEven}).
We exhibit a family of specular sets obtained as the result of a transformation called doubling, starting from a tree set of characteristic $1$ and invariant by reversal (Theorem~\ref{theoremDoubling}).
In the last part, we relate specular sets with full and $G$-full words, a notion linked with palindromic complexity and introduced in~\cite{PelantovaStarosta2014}.

\subsection{Definition}

We assume given an involution $\theta$ on the alphabet $A$ generating the specular group $G_\theta$.

A symmetric biextendable (and thus factorial)
set $S$ of reduced words on the alphabet $A$
is called
a \emph{laminary set} on $A$ relative to $\theta$ (following~\cite{HilionCoulboisLustig2008}
and~\cite{LopezNarbel2013}). Thus the elements of a laminary set $S$
are elements of the specular group $G_\theta$ and the set $S$ is contained
in $G_\theta$.

A \emph{specular set} is a laminary set on $A$ which is a tree set of characteristic $2$. Thus, in a specular set, the extension graph of every nonempty word
is a tree and the extension graph of the empty word is a union of two
disjoint trees.

The following is a very simple example of a specular set.
\begin{example}\label{exampleabab}
Let $A=\{a,b\}$ and let $\theta$ be the identity on $A$. Then the set
of factors of $(ab)^\omega$ is a specular set
(we denote by $x^\omega$ the word $x$ infinitely repeated).
\end{example}
The next example is due to Julien Cassaigne. 
We frequently refer to it in next sections.

\begin{example}\label{exampleJulien}
Let $A=\{a,b,c,d\}$ and let $S$ be
the set of factors of the fixed point $\sigma^\omega(a)$
of the morphism $\sigma$ from $A^*$ into
itself defined by 
\begin{displaymath}
\sigma(a)=ab,\quad \sigma(b)=cda,\quad \sigma(c)=cd,\quad \sigma(d)=abc.
\end{displaymath}
\begin{figure}[hbt]
\centering\gasset{Nadjust=wh,AHnb=0}
\begin{picture}(50,10)
\put(0,0){
\begin{picture}(20,10)
\node(bl)(0,0){$d$}\node(al)(0,10){$c$}
\node(cr)(20,0){$a$}\node(br)(20,10){$d$}

\drawedge(bl,cr){}\drawedge(al,cr){}
\drawedge(al,br){}
\end{picture}
}
\put(30,0){
\begin{picture}(20,10)
\node(dl)(0,0){$b$}\node(cl)(0,10){$a$}
\node(ar)(20,0){$c$}\node(dr)(20,10){$b$}

\drawedge(bl,cr){}\drawedge(al,cr){}
\drawedge(al,br){}
\end{picture}
}
\end{picture}
\caption{The extension graph $\E_S(\varepsilon)$.}\label{figureExtensionGraph}
\end{figure}

The extension graph of $\varepsilon$ is shown in Figure~\ref{figureExtensionGraph}.
It is shown in~\cite[Example 3.4]{BertheDeFeliceDolceLeroyPerrinReutenauerRindone2013a} that $S$ is a tree set
of characteristic $2$. We will see
later (Example~\ref{exampleJulien2}) that $S$ is a specular set relative to the involution $\theta=(bd)$.
\end{example}
The following result shows in particular
 that in a specular set the two trees forming $\E(\varepsilon)$
are isomorphic since they are exchanged by the bijection $(a,b)\rightarrow (b^{-1},a^{-1})$. 
\begin{proposition}\label{propositionClelia}
Let $S$ be a specular set.
Let $\T_0,\T_1$ be the two trees such that $\E(\varepsilon)=\T_0\cup \T_1$.
For any $a,b\in A$ and $i=0,1$, one has $(1\otimes a,b\otimes 1)\in \T_i$
if and only if  $(1\otimes b^{-1},a^{-1}\otimes 1)\in \T_{1-i}$.
\end{proposition}
\begin{proof}
Assume that $(1\otimes a,b\otimes 1)$
and  $(1\otimes b^{-1},a^{-1}\otimes 1)$ are both in $\T_0$. Since $\T_0$
is a tree, there is a path from $1\otimes a$ to $a^{-1}\otimes 1$. 
We may assume that this path is reduced, that is, does not use
consecutively twice the same edge. Since this path is of odd
length, it has the form $(u_0,v_1,u_1,
\ldots,u_p,v_p)$ with $u_0=1\otimes a$ and $v_p=a^{-1}\otimes 1$.
Since $S$ is symmetric, we also have a reduced path $(v_p^{-1},u_p^{-1},\cdots,u_1^{-1},u_0^{-1})$ which is in $\E(\varepsilon)$ (for $u_i=1\otimes a_i$, we denote $u_i^{-1}=a_i^{-1}\otimes 1$ and similarly for $v_i^{-1}$) and thus in $\T_0$ since $\T_0$ and $\T_1$ are disjoint.
Since $v_p^{-1}=u_0$, these two paths have the same origin and end.
But if a path of odd length is its own inverse, its central
edge has the form $(x,y)$ with $x=y^{-1}$, as one verifies easily
by induction on the length of the path. This is a contradiction with
the fact that the words of $S$ are reduced.
Thus the two paths  are distinct.
This implies that $\E(\varepsilon)$ has
a cycle, a contradiction.
\end{proof}

Following again the terminology of~\cite{HilionCoulboisLustig2008}, we say
that a laminary set $S$ is \emph{orientable} if there exist two factorial
sets $S_+,S_-$ such that $S=S_+\cup S_-$ with $S_+\cap S_-=\{\varepsilon\}$
and for any $x\in S$, one has $x\in S_-$ if and only if $x^{-1}\in S_+$
(where $x^{-1}$ is the inverse of $x$ in $G_\theta$).

The following result shows in particular that for 
any tree set $T$ of characteristic 1
on the alphabet $B$, the set $T\cup T^{-1}$
is a specular set on the alphabet $A=B\cup B^{-1}$.

\begin{theorem}
Let $S$ be a specular set on the alphabet $A$.
Then, $S$ is orientable if and only if there is a partition
$A = A_{+} \cup A_{-}$ of the alphabet $A$ and a tree set $T$ of characteristic 1 on the alphabet $B = A_{+}$ such that $S = T \cup T^{-1}$.
\end{theorem}

\begin{proof}
The condition is trivially sufficient.
Let us prove it is necessary and suppose that $S$ is a specular set on the alphabet $A$ which is orientable.
Let $(S_+,S_-)$ be the corresponding pair of subsets of $S$.
 The sets $S_+,S_-$ are biextendable, since $S$ is.
Set $A_+=A\cap S_+$ and $A_-=A\cap S_-$. Then $A=A_+\cup A_-$ is a partition
of $A$ and, since $S_-,S_+$ are factorial, we have $S_+\subset A_+^*$
and $S_-\subset A_-^*$.
Let $\T_0,\T_1$ be the two trees such that $\E(\varepsilon)=\T_0\cup \T_1$.
Assume that a vertex of $\T_0$ is in $A_+$. Then all vertices
of $\T_0$ are in $A_+$ and all vertices of $\T_1$ are in $A_-$.
Moreover, $\E_{S_+}(\varepsilon)=\T_0$ and $\E_{S_-}(\varepsilon)=\T_1$.
Thus $S_+,S_-$ are tree sets of characteristic 1.
\end{proof} 

The following result follows easily from Proposition~\ref{propCANT} (see~\cite[Proposition 2.4]{DolcePerrin2016} for details).

\begin{proposition}\label{propComplexity}
The factor complexity of a specular set containing the alphabet $A$ is given by $p_0=1$ and $p_n=n(k-2)+2$ for $n\ge 1$ with $k=\Card(A)$. 
\end{proposition}

\subsection{Odd and even words}\label{sectionOddEven}
We introduce a notion which plays, as  we shall see, an important role in the study of specular sets.
Let $S$ be a specular set.
Since a specular set is biextendable, any letter $a\in A$ occurs exactly twice as a vertex of $\E(\varepsilon)$, one as an element of $L(\varepsilon)$ and one as an element of $R(\varepsilon)$.
A letter $a\in A$ is said to be \emph{even} if its two occurrences appear in the same tree.
Otherwise, it is said to be \emph{odd}.
Observe that if a specular $S$ is recurrent, there is at least one odd letter. 

\begin{example}
Let $S$ be the set of factors of $(ab)^\omega$ as in Example~\ref{exampleabab}.
Then $a$ and $b$ are odd.
\end{example}

\begin{example}
Let $S$ be the set of Example~\ref{exampleJulien}.
The letters $b,d$ are even, while $a$ and $c$ are odd.
\end{example}

Let $S$ be a specular set.
A word $w\in S$ is said to be \emph{even} if it has an even number
of odd letters. Otherwise it is said to be \emph{odd}.
The set of even words has the form $X^*\cap S$ where $X\subset S$
 is a bifix code, called the \emph{even code}. The set $X$
is the set of even words without a nonempty even
prefix (or suffix). 
\begin{proposition}\label{propositionEvenCode}
 Let $S$ be a recurrent specular set.
The even code is an $S$-maximal bifix code of $S$-degree $2$.
\end{proposition}
\begin{proof} 
Let us verify that
any $w\in S$ is comparable for the prefix order with an element of
the even code
$X$. If $w$ is even, it is in $X^*$. Otherwise, since $S$ is recurrent,
there is a word $u$ such that $wuw\in S$. If $u$ is even, then
$wuw$ is even and thus $wuw\in X^*$. Otherwise $wu$ is even and
thus $wu\in X^*$. This shows that $X$ is $S$-maximal. 
 The fact
that it has $S$-degree $2$ follows from the fact that
any product of two odd letters is a word of $X$ which is not an
internal factor of $X$ and has two parses.
\end{proof}

\begin{example}\label{exampleEvenCode}
Let $S$ be the specular set of Example~\ref{exampleJulien}. The letters
$b,d$ are even and the letters $a,c$ are odd. The even code is
\begin{displaymath}
X=\{abc,ac,b,ca,cda,d\}.
\end{displaymath}
\end{example}

Denote by $\T_0,\T_1$ the two trees such that $\E(\varepsilon)=\T_0\cup \T_1$.
We consider the directed graph $\G$ with vertices $0,1$ and edges all the triples
$(i,a,j)$ for $0\le i,j\le 1$ and $a\in A$ such that $(1\otimes b,a\otimes 1)\in \T_i$ and
$(1\otimes a,c\otimes 1)\in \T_j$  for some $b,c\in A$. 
The graph $\G$ is called the
\emph{parity graph} of $S$. Observe that for every letter $a\in A$ there
is exactly one edge labeled $a$ because $a$ appears exactly once as a left
(resp. right) vertex in $\E(\varepsilon)$.

\begin{example}\label{exampleparitygraph}
Let $S$ be the specular set of Example~\ref{exampleJulien}. 
The parity graph of $S$ is represented in Figure~\ref{figureParityGraph}, where we assume that $\mathcal{T}_0$ is the tree on the left of Figure~\ref{figureExtensionGraph} and $\mathcal{T}_1$ is the tree on the right of Figure~\ref{figureExtensionGraph}.
\begin{figure}[hbt]
\centering\gasset{Nadjust=wh}
\begin{picture}(20,10)
\node(0)(0,5){$0$}\node(1)(20,5){$1$}

\drawloop[loopangle=180](0){$d$}
\drawedge[curvedepth=3](0,1){$a$}\drawedge[curvedepth=3](1,0){$c$}
\drawloop[loopangle=0](1){$b$}
\end{picture}
\caption{The parity graph.}\label{figureParityGraph}
\end{figure}
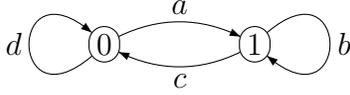
\end{example}

\begin{proposition}\label{propositionPartition}
Let $S$ be a  specular set and let $\G$ be its parity graph.
Let $S_{i,j}$ be the set of words in $S$ which are the label of a path
from $i$ to $j$ in the graph $\G$. 
\begin{enumerate}
\item[(1)] The family $(S_{i,j}\setminus\{\varepsilon\})_{0\le i,j\le 1}$ is a partition of $S\setminus \{\varepsilon\}$.
\item[(2)] For $u\in S_{i,j}\setminus\{\varepsilon\}$ and $v\in S_{k,\ell}\setminus\{\varepsilon\}$, if $uv\in S$,
then $j=k$. 
\item[(3)] $S_{0,0}\cup S_{1,1}$ is the set of even words.
\item[(4)] $S_{i,j}^{-1}=S_{1-j,1-i}$.
\end{enumerate}
\end{proposition}
\begin{proof}
We first note that for $a,b\in A$ such that $ab\in S$, there is a path
in $\G$ labeled $ab$. Since $(a,b)\in \E(\varepsilon)$,
there is a $k$  such that $(1\otimes a,b\otimes 1)\in\T_k$. Then we have
$a\in S_{i,k}$ and $b\in S_{k,j}$ for some $i,j\in\{0,1\}$. This shows
that $ab$ is the label of a path from $i$ to $j$ in $\G$.

Let us prove by induction on the length of a nonempty word $w\in S$ 
that there exists a unique pair $i,j$ such that $w\in S_{i,j}$.
The property is true for a letter, by definition of the extension
graph $\E(\varepsilon)$ and for words of length $2$
by the above argument. Let next $w=ax$ be in $S$ with $a\in A$
and $x$ nonempty. By induction hypothesis, there is a unique pair
$(k,j)$ such that $x\in S_{k,j}$. Let $b$ be the first letter
of $x$. Then the edge of $\G$ with label $b$ starts in $k$. Since
$ab$ is the label of a path, we have $a\in S_{i,k}$ for some
$i$ and thus $ax\in S_{i,j}$.
The other assertions follow easily
(Assertion (4) follows from Proposition~\ref{propositionClelia}).
\end{proof}
Note that Assertion (4) implies that no nonempty even word is its own
inverse. Indeed, $S_{0,0}^{-1}=S_{1,1}$ and $S_{1,1}^{-1}=S_{0,0}$.

\begin{proposition}\label{propositionOddReturns}
Let $S$ be a specular set. 
If $x,y\in S$ are nonempty words such that $xyx^{-1}\in S$, then $y$ is odd.
\end{proposition}
\begin{proof}
Let $i,j$ be such that $x\in S_{i,j}$. Then $x^{-1}\in S_{1-j,1-i}$ by 
Assertion (4) of
Proposition~\ref{propositionPartition} and thus $y\in S_{j,1-j}$ by Assertion
 (2).
 Thus $y$ is odd by Assertion (3).
\end{proof}

The following result is the counterpart for recurrent specular sets of the main result of~\cite[Theorem 6.1]{BertheDeFeliceDolceLeroyPerrinReutenauerRindone2013m} asserting that the family of (uniformly) recurrent tree sets of characteristic $1$ is closed under maximal bifix decoding.
Let $S$ be a recurrent set and let $f$ be a coding morphism for a finite $S$-maximal bifix code $X$. The set $f^{-1}(S)$ is called a \emph{decoding} of $S$ by $X$.

\begin{theorem}[Even code decoding Theorem]
\label{theoremDecodingEven}
The decoding of a  recurrent specular set by the even code is  a union of two recurrent tree sets of characteristic $1$.
More precisely, let $S$ be a recurrent specular set and let $f$ be a coding morphism for the even code.
Then $f^{-1}(S_{0,0})$ and $f^{-1}(S_{1,1})$ are recurrent tree sets of characteristic $1$.
\end{theorem}
\begin{proof}
We show that $T_0=f^{-1}(S_{0,0})$ is a recurrent tree set of characteristic $1$.
The proof for $f^{-1}(S_{1,1})$ is the same.

First, $T_0$ is biextendable, as one may easily verify.
Next, since $S$ is recurrent, it is uniformly recurrent by Theorem~\ref{theo:recurrentur}.
Thus for every $u\in S$ there exists $n\ge 1$ such that $u$ is a factor of any word $w$ in $S$ of length $n$.
But if $u,w\in S_{0,0}$ are such that $w=\ell ur$, then $\ell,r\in S_{0,0}$.
Thus $T_0$ is recurrent.

We now show that $T_0$ is a tree set of characteristic $1$.
Let $X$ be the even code and set $X_0=X\cap S_{0,0}$ and $X_1=X\cap S_{1,1}$.

It is enough to show that $\E_S^0(w) = \E_S^{X_0,X_0}(w)$ is a tree for any $w\in S_{0,0}$.
Note first that $\E_S^0(w) = \E_S^{X,X}(w)$.
Indeed, for $w\in S_{0,0}$ and $x,y\in X$ such that $xwy\in S$, one has $x,y\in X_0$ and $xwy\in S_{0,0}$.

First, for any nonempty word $w\in S_{0,0}$, since $\E_S^0(w) = \E_S^{X,X}(w)$, the graph $\E_S^0(w)$ is a tree by Proposition~\ref{propositionExtended}.

Next, let us show that the graph $\E_S^0(\varepsilon)$ is a tree.
First, since $\E_S(\varepsilon)$ is a union of two trees, it is acyclic, and thus the graph $\E_S^0(\varepsilon)$ is acyclic by Proposition~\ref{propositionExtended}.
Next, since $S$ is neutral, by Proposition~\ref{propositionExtended}, we have $m_S^{X,X}(\varepsilon)=m_S(\varepsilon)=-1$.
This implies that $m_S^{X,X}(\varepsilon)$ is a union of two trees.
Since $\E_S^{X,X}(\varepsilon)$ is the disjoint union of $\E_S^0(\varepsilon)$ and $\E_S^{X_1,X_1}(\varepsilon)$, this implies that each one is a tree.
\end{proof}

\begin{example}
Let $S$ be the set of Example~\ref{exampleJulien}. 
Recall that it is the set of
factors of the fixed point of the morphism
\begin{displaymath}
\sigma:a\mapsto ab,\quad b\mapsto cda,\quad c\mapsto cd,\quad d\mapsto abc.
\end{displaymath}
The even code $X$
is given in Example~\ref{exampleEvenCode}. Let $\Sigma=\{a,b,c,d,e,f\}$
and let $g$ be the
coding morphism for $X$ given by 
\begin{displaymath}
a\mapsto abc,\quad b\mapsto ac,\quad c\mapsto b,\quad d\mapsto ca,\quad e\mapsto cda,\quad f\mapsto d.
\end{displaymath}
The decoding of $S$
by $X$ is a union of two tree sets of characteristic $1$ which are the set of factors of the
fixed point of the two morphisms
\begin{displaymath}
a\mapsto afbf,\ b\mapsto af,\ f\mapsto a
\end{displaymath}
and 
\begin{displaymath}
c\mapsto e,\ d\mapsto ec,\ e\mapsto ecdc.
\end{displaymath}
These two morphisms are actually the restrictions to $\{a,b,f\}$ and $\{c,d,e\}$
of
the morphism $g^{-1}\sigma g$.
\end{example}

\subsection{Bifix codes in specular sets}
Recall from Section~\ref{sectionPreliminaries} that the characteristic of a set $S$ is given by $\chi(S) = \ell_S(\varepsilon) + r_S(\varepsilon) - b_S(\varepsilon)$.

The following result is from~\cite{DolcePerrin2016}.
We will use it for specular sets.

\begin{theorem}
\label{theoremCardinality}
Let $S$ be a recurrent neutral set containing the alphabet $A$.
For any finite $S$-maximal bifix code $X$ of $S$-degree $d=d_X(S)$, one has
\begin{displaymath}
\Card(X)=d(\Card(A)-\chi(S))+\chi(S).
\end{displaymath}
\end{theorem}

We can apply Theorem~\ref{theoremCardinality} to recurrent specular sets.

\begin{theorem}[Cardinality Theorem for bifix codes]
\label{corollaryCardinality}
Let $S$ be a recurrent specular set containing the alphabet $A$.
For any finite $S$-maximal bifix code $X$, one has 
\begin{equation}
\Card(X)=d_X(S)(\Card(A)-2)+2.
\label{formulaBifixMax}
\end{equation}
\end{theorem}
\begin{proof}
Since $S$ is specular, we have $\chi(S)=2$ and thus the statement follows directly from Theorem~\ref{theoremCardinality}.
\end{proof}
\begin{example}
Let $S$ be the specular set of Example~\ref{exampleJulien}.
The even code (given in Example~\ref{exampleEvenCode}) is an $S$-maximal code of $S$-degree $2$.
We have $\Card(X)=6$ in agreement with Theorem~\ref{theoremCardinality}.
\end{example}

The following statement is a partial converse of Theorem~\ref{theoremCardinality}. 

\begin{theorem}
\label{theoremConverseCard}
Let $S$ be a uniformly recurrent laminary set containing the alphabet $A$.
If the graph $\E(\varepsilon)$ is acyclic and if any finite $S$-maximal bifix code of $S$-degree $d$ has $d(Card(A)-2) +2$ elements, then $S$ is specular.
\end{theorem}

To prove Theorem~\ref{theoremConverseCard}, we use the following result, which can be proved in the same way as Theorem 3.12 in~\cite{BertheDeFeliceDolceLeroyPerrinReutenauerRindone2013b}, using internal transformations.

\begin{proposition}
\label{propositionConverseCard}
Let $S$ be a recurrent set containing the alphabet $A$ and let $d_0\ge 2$.
If all finite $S$-maximal bifix codes of $S$-degree $d \geq d_0$ have the same cardinality, then any word of length greater than or equal to $d_0-1$ is neutral.
\end{proposition}

Theorem~\ref{theoremConverseCard} results from Proposition~\ref{propositionConverseCard} applied with $d_0=2$.

\subsection{Doubling maps}
\label{sec:doubling}
We now introduce a construction which allows one to build specular sets.
This is a particular case of the multiplying maps introduced in~\cite{DolcePerrin2016}.

A \emph{transducer} is a labeled graph with vertices
in a set $Q$ and edges labeled in $\Sigma\times A$. The set $Q$ is called the set of states, the set $\Sigma$ is called the \emph{input alphabet}
and $A$ is called the \emph{output alphabet}.
The graph obtained by erasing the output
letters is called the \emph{input automaton} (with an unspecified initial state). Similarly, the \emph{output automaton} is obtained by erasing the input letters.

Let $\A$ be a transducer with set of states $Q=\{0,1\}$ on the input alphabet
$\Sigma$ and the output alphabet $A$. We assume that
\begin{enumerate}
\item the input automaton is a group automaton,
 that is, every
letter of $\Sigma$ acts on $Q$ as a permutation
\item the output labels of the edges are all distinct.
\end{enumerate}
  We define two maps $\delta_0,\delta_1:\Sigma^*\rightarrow A^*$ corresponding to initial states
$0$ and $1$ respectively. 
Thus $\delta_0(u)=v$ (resp. $\delta_1(u)=v$) if the path starting at state $0$ (resp. $1$) with input label $u$ has output $v$.
The pair $\delta_\A = (\delta_0,\delta_1)$ is called a \emph{doubling map}
and the transducer $\A$ a \emph{doubling transducer}.

The \emph{image} of a set $T$ on the alphabet $\Sigma$
by the doubling map $\delta_\A$ is the set $S=\delta_0(T)\cup \delta_1(T)$.

If $\A$ is a doubling transducer, we define an involution $\theta_\A$
as follows. For any $a\in A$, let $(i,\alpha,a,j)$ be the edge with input
label $\alpha$ and output label $a$. We define $\theta_\A(a)$ as the output
label of the edge starting at $1-j$ with input label $\alpha$. Thus, $\theta_\A(a)=\delta_i(\alpha)=a$
if $i+j=1$ and $\theta_\A(a)=\delta_{1-i}(\alpha)\ne a$ if $i=j$.

Recall that the reversal of a word $w=a_1a_2\cdots a_n$ is the word $\tilde{w}=a_n\cdots a_2a_1$.

One can prove by induction on the length of $y \in \Sigma^*$ that if $x = \delta_i(y)$ and if $j$ is the end of the path starting at $i$ and with input label $y$, then $x^{-1} = \delta_{1-j}(\tilde{y})$.
Observe that since the input automaton is a group automaton, there is always a path starting at $1-j$ with input label $\tilde{y}$.

A set $S$ of words is closed under reversal if $w\in S$ implies $\tilde{w}\in S$ for every $w\in S$.
By definition, any Sturmian set is closed under reversal (see~\cite{BerstelDeFelicePerrinReutenauerRindone2012}).

\begin{theorem}\label{theoremDoubling}
For any tree set $T$ of characteristic $1$
on the alphabet $\Sigma$, closed under reversal
and any doubling map $\delta_\A$, the image of $T$ by $\delta_\A$
  is a specular set relative to the
involution $\theta_\A$.
\end{theorem}
\begin{proof}
Set $S=\delta_\A(T)=\delta_0(T)\cup \delta_1(T)$.
By Theorem 3.1 of~\cite{DolcePerrin2016}, $S$ is a tree set of characteristic 2.
By construction, it is also clear the any word in $S$ is $\theta_\A$-reduced.

Let now prove that $S$ is a symmetric language.
Assume that  $x=\delta_i(y)$ for $i\in\{0,1\}$ and $y\in T$.
Let $j$ be the end of the path starting at $i$ and with input
label $y$. 
Since $x^{-1}=\delta_{1-j}(\tilde{y})$ and $T$ is closed under reversal,
we have $x^{-1}\in\delta_{1-j}(T)$.
This shows that $S$ is symmetric and so that it is laminary.
Thus, $S$ is a specular set.
\end{proof}

We now give two examples of specular sets obtained by doubling maps (doubling the Fibonacci set).

\begin{example}
\label{ex:fibodouble2}
Let $\Sigma = \{ \alpha, \beta \}$ and let $T$ be the Fibonacci set over $\Sigma$.
Let $\delta$ be the doubling map given by the transducer of Figure~\ref{fig:fibodouble2} on the left.

\begin{figure}[hbt]
\centering\gasset{Nadjust=wh}
\begin{picture}(100,15)
\put(0,0){
\begin{picture}(50,15)
\node(0)(0,7){$0$}
\node(1)(30,7){$1$}

\drawloop[loopangle=30](0){$\alpha\mid a$}
\drawloop[loopangle=210](0){$\beta\mid b$}
\drawloop[loopangle=30](1){$\alpha\mid c$}
\drawloop[loopangle=210](1){$\beta\mid d$}
\end{picture}
}
\gasset{AHnb=0}
\put(50,0){
\begin{picture}(50,10)
\put(0,0){
\begin{picture}(20,10)
\node(dl)(0,0){$b$}
\node(cl)(0,10){$a$}
\node(ar)(20,0){$a$}
\node(dr)(20,10){$b$}

\drawedge(bl,cr){}
\drawedge(al,cr){}
\drawedge(al,br){}
\end{picture}
}
\put(30,0){
\begin{picture}(20,10)
\node(bl)(0,0){$d$}
\node(al)(0,10){$c$}
\node(cr)(20,0){$c$}
\node(br)(20,10){$d$}

\drawedge(bl,cr){}
\drawedge(al,cr){}
\drawedge(al,br){}
\end{picture}
}
\end{picture}
}
\end{picture}
\caption{A doubling transducer and the extension graph $\E_S(\varepsilon)$.}
\label{fig:fibodouble2}
\end{figure}

Both letters in $\Sigma$ act as the identity on the two states $0,1$.

Then $\theta_\A$ is the involution defined by $\theta : a \mapsto c, \ b \mapsto d, \ c \mapsto a, \ d \mapsto b$.
The image of $T$ by $\delta$ is a specular set $S$ on the alphabet $A=\{a,b,c,d\}$.
The graph $\E_S(\varepsilon)$ is represented in Figure~\ref{figureFiboDouble} on the right.
All letters are even.

Note that the set $S$ of Example~\ref{ex:fibodouble2} is not recurrent.
The set $S$ is actually just a union of two Fibonacci sets, one over the alphabet $\{ a,b \}$ and the second over the alphabet $\{ c,d \}$.
\end{example}

\begin{example}
\label{exampleFiboDouble}
Let $\Sigma=\{\alpha,\beta\}$ and let $T$ be the Fibonacci set.
Let $\delta$ be the doubling map given by the transducer of Figure~\ref{figureFiboDouble} on the left.
The letter $\alpha$ acts as the transposition of the two states $0,1$, while $\beta$ acts as the identity.

\begin{figure}[hbt]
\centering
\gasset{Nadjust=wh}
\begin{picture}(100,12)
\put(5,0){
\begin{picture}(20,10)
\node(0)(0,5){$0$}\node(1)(20,5){$1$}

\drawloop[loopangle=180](0){$\beta\mid d$}
\drawedge[curvedepth=3](0,1){$\alpha\mid a$}\drawedge[curvedepth=3](1,0){$\alpha\mid c$}
\drawloop[loopangle=0](1){$\beta\mid b$}
\end{picture}
}
\gasset{AHnb=0}
\put(50,0){
\begin{picture}(50,10)
\put(30,0){
\begin{picture}(20,10)
\node(bl)(0,0){$b$}\node(al)(0,10){$a$}
\node(cr)(20,0){$c$}\node(br)(20,10){$b$}

\drawedge(bl,cr){}\drawedge(al,cr){}
\drawedge(al,br){}
\end{picture}
}
\put(0,0){
\begin{picture}(20,10)
\node(dl)(0,0){$d$}\node(cl)(0,10){$c$}
\node(ar)(20,0){$a$}\node(dr)(20,10){$d$}

\drawedge(bl,cr){}\drawedge(al,cr){}
\drawedge(al,br){}
\end{picture}
}
\end{picture}
}
\end{picture}
\caption{A doubling transducer and the extension graph $\E_S(\varepsilon)$.}
\label{figureFiboDouble}
\end{figure}
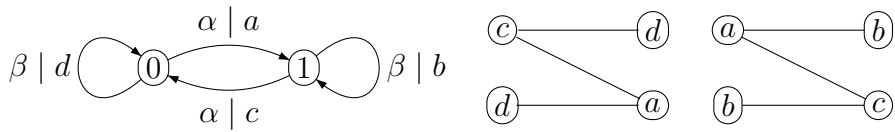

Then $\theta_\A$ is the involution $\theta$ of Example~\ref{exampleSpecularGroup} and the image of $T$ by $\delta$ is a specular set $S$
on the alphabet $A=\{a,b,c,d\}$.
The graph $\E_S(\varepsilon)$ is represented
in Figure~\ref{figureFiboDouble} on the right.

The letters $a,c$ are odd and $b,d$ are even.

Note that $S$ is the set of factors of the fixed point $g^\omega(a)$
of the morphism
\begin{displaymath}
g: a\mapsto abcab,\quad b\mapsto cda,\quad c\mapsto cdacd,\quad d\mapsto abc.
\end{displaymath}
The morphism $g$ is obtained by applying the doubling map
to the cube $f^3$ of the Fibonacci morphism $f$ in such a way that
$g^\omega(a)=\delta_0(f^\omega(\alpha))$.
\end{example}

In the next example (due to Julien Cassaigne), the specular set
is obtained using a morphism of smaller size.
\begin{example}\label{exampleJulien2}
Let $A=\{a,b,c,d\}$.
Let $T$ be the set of factors of the fixed point $x=f^\omega(\alpha)$ of
the morphism $f:\alpha\mapsto \alpha\beta,\beta\mapsto \alpha\beta\alpha$. It is a Sturmian set.
Indeed, $x$ is the characteristic word of slope $-1+\sqrt{2}$
(see~\cite{Lothaire2002}).
The sequence $s_n=f^n(\alpha)$ satisfies $s_n=s_{n-1}^2s_{n-2}$ for $n\ge 2$.
The image $S$ of $T$ by
the doubling  automaton of Figure~\ref{figureFiboDouble} is
the set of factors of the fixed point $\sigma^\omega(a)$
of the morphism $\sigma$ from $A^*$ into
itself defined by 
\begin{displaymath}
\sigma(a)=ab,\quad \sigma(b)=cda,\quad \sigma(c)=cd,\quad \sigma(d)=abc.
\end{displaymath}
Thus the set $S$ is the same as that of Example~\ref{exampleJulien}.
\end{example}

Note that, when $S$ is a specular set obtained by a doubling map using a transducer $\A$, the parity graph of $S$ is the output automaton of $\A$ (see for instance Figures~\ref{figureParityGraph} and~\ref{figureFiboDouble}).

\subsection{Palindromes}
\label{sec:palindromes}

The notion of palindromic complexity originates in~\cite{DroubayJustinPirillo2001} where it is proved that a word of length $n$ has at most $n+1$ palindrome factors.
A word of length $n$ is full if it has $n+1$ palindrome factors and a factorial set is \emph{full} (or rich) if all its elements are full.
By a result of~\cite{GlenJustinWidmerZamboni2009}, a recurrent set closed under reversal is full if and only if every complete return word to a palindrome in $S$ is a palindrome (a complete return word to a set $X$ of words of the same length is a word of $S$ which has exactly two factors in $X$, one as a proper prefix and one as a proper suffix, see Section~\ref{subsec:cardret}).
It is known that all Sturmian sets are full~\cite{DroubayJustinPirillo2001} and also all natural codings of interval exchange defined by a symmetric permutation \cite{BalaziMaskovaPelantova2007}.

The fact that a tree set of characteristic $1$ is full in the following result generalizes results of~\cite{DroubayJustinPirillo2001,BalaziMaskovaPelantova2007}.

\begin{proposition}
\label{propositionfull}
Let $T$ be a recurrent tree set of characteristic $1$, closed under reversal.
Then $T$ is full.
\end{proposition}
\begin{proof}
We use the following equivalent definition of full sets (see~\cite{PelantovaStarosta2014}): for any $x\in T$,
\begin{enumerate}
\item[(i)]
if $x$ is not a palindrome, it is neutral. 
\item[(ii)] Otherwise, $m(x)+1$ is equal to the number of letters $a$ such that $axa$ is a palindrome in $T$ (the so-called \emph{palindromic extensions}).
\end{enumerate}
Since $T$ is a tree set of characteristic $1$, every word is neutral.
We thus only have to show that every palindrome has exactly one palindromic extension.
Let $x \in T$ be a palindrome.
It may be verified that since $x$ is palindrome and $T$ is closed under reversal, the graph $\E_T(x)$ is closed under reversal in the sense that it contains an edge $(1 \otimes a, b \otimes 1)$ if and only if it contains the edge $(1 \otimes b, a \otimes 1)$.
One may verify that, as a consequence, there is at least one $a \in A$ such that $axa \in T$.
Indeed, this can be proved as follows by induction on $\Card(A)$.
It is true if $\Card(A)=1$.
Otherwise, let $a \in A$ be such that $1 \otimes a$ is a leaf of $\E_T(x)$.
Then, since the graph is closed under reversal, the vertex $a \otimes 1$ is also a leaf.
Set $A' = A \setminus \{ a \}$.
The restriction of the graph to the vertices in $A'$ is a tree closed under reversal, and thus the property follows by induction.
But if there is another one, the graph would have a cycle.
Indeed, assume that $axa, bxb \in T$.
Consider a simple path $\gamma$ of minimal length from one of $1 \otimes a, a \otimes 1$ to one of $1\otimes b, b \otimes 1$.
This path cannot contain the edges corresponding to $axa, bxb$.
Using these edges and the symmetric of $\gamma$, one obtains a cycle.
Thus $T$ is full.
\end{proof}

In~\cite{PelantovaStarosta2014}, this notion was extended to that of $G$-full, where $G$ is a finite group of morphisms and antimorphisms of $A^*$ (an antimorphism is the composition of a morphism
and reversal) containing at least one antimorphism.
As one of the equivalent definitions, a set $S$ closed under $G$ is $G$-full if for every $x \in S$, every complete return word to the $G$-orbit of $x$ is fixed by a nontrivial element of $G$.

Let us consider a tree set $T$ of characteristic $1$ and a specular set $S$ obtained as the image of $T$ by a doubling map $\delta$.

Let us define the antimorphism $\sigma: u \mapsto u^{-1}$ for $u\in G_\theta$.
From Section~\ref{sec:doubling} it follows that both edges $(i, \alpha, a, j)$ and $(1-i, \alpha, \sigma(a), 1-j)$ are in the doubling transducer.
Let us define also the morphism $\tau$ obtained by replacing each letter $a\in A$ by $\tau(a)$ if there are edges $(i,\alpha,a,j)$ and $(1-j,\alpha,\tau(a),1-i)$ in the doubling transducer.

We denote by $G_\mathcal{A}$ the group generated by the $\sigma$ and $\tau$.
Actually, we have $G_\mathcal{A} = \Z/2\Z \times \Z/2\Z$.

\begin{example}
\label{ex:h2}
Let $S$ be the specular set defined in Example~\ref{ex:fibodouble2}.
The group $G_\mathcal{A}$ is generated by
$$\sigma : a t\mapsto c, \ b \mapsto d, \ c \mapsto a, \ d \mapsto b,$$
and
$$\tau : a \mapsto c, \ b \mapsto d, \ c \mapsto a, \ d \mapsto b.$$
Note that, even if the images of $\sigma$ and $\tau$ over the alphabet are the same, the latter is a morphism, while the first is an antimorphism.
Moreover, in that case, we have $\sigma \tau = \tau \sigma : w \mapsto \tilde{w}$ for every $w \in S$.
\end{example}

\begin{example}
\label{ex:h}
Let $S$ be the recurrent specular set defined in Example~\ref{exampleFiboDouble}.
The group $G_\mathcal{A}$ is generated by the antimorphism
$$\sigma : a \mapsto a, \ b \mapsto d, \ c \mapsto c, \ d \mapsto a,$$
and the morphism
$$\tau : a \mapsto c, \ b \mapsto d, \ c \mapsto a, \ d \mapsto b.$$
We have $H = \{ \id, \sigma, \tau, \sigma \tau \}$, where $\sigma \tau = \tau \sigma$ is the antimorphism fixing $b,d$ and exchanging $a$ and $c$.

\end{example}

We now connect the notions of fullness and $G_\mathcal{A}$-fullness, proving an analogous result of Proposition~\ref{propositionfull} for specular sets.

\begin{proposition}
\label{pro:Hfull}
Let $T$ be a recurrent tree set of characteristic $1$ on the alphabet $\Sigma$, closed under reversal and let $S$ be the image of $T$ under a doubling map.
Then $S$ is $G_\mathcal{A}$-full.
\end{proposition}
\begin{proof}
By Proposition~\ref{propositionfull} we know that $T$ is full.

To show that $S$ is $G_\mathcal{A}$-full, we will use several properties of the map $\delta_i$.
We note that it is injective, that it preserves prefixes and conversely: $u$ is a prefix of $v$ if and only if $\delta_i(u)$ is a prefix of $\delta_i(v)$.
Also, for any $y \in T$ and $x=\delta_i(y)$, the images of $y,\tilde{y}$ by $\delta_0,\delta_1$ form the $G_\mathcal{A}$-orbit of $x$.

Consider $x\in S$ and a word $w$ which is a complete return word to the $G_\mathcal{A}$-orbit of $x$.
We may assume that $x$ is a prefix of $w$ and that $\gamma(x)$ is a prefix of $w$, with $\gamma \in H$.
Let $y,u \in T$ and $i \in \{0,1\}$ be such that $x=\delta_i(y)$ and $w=\delta_i(u)$.
Then $y$ is a prefix of $u$.

We first show that $u$ is a palindrome.
First observe that $u$ has a suffix in the set $\{ y, \tilde{y} \}$.
Indeed, if $\gamma \in \{ \id, \tau \}$ then $y$ is a suffix of $u$. Otherwise, if $\gamma \in \{ \sigma, \tau \sigma \}$, one has that $\tilde{y}$ is a suffix of $u$.
Let now $z$ be the longest palindrome prefix of $u$.
Then $y$ is a prefix of $z$ since otherwise $z$ would have a second occurrence in $u$ (in a full set, the longest palindrome prefix of a word is unioccurrent, see~\cite{GlenJustinWidmerZamboni2009}).
Consequently $\tilde{y}$ is a suffix of $z$ and $z$ cannot have another occurrence of $y$ or $\tilde{y}$ except as a prefix or a suffix (otherwise, $w$ would have an internal factor in the $G_\mathcal{A}$-orbit of $x$).
Thus $z$ is a complete return word to $\{y,\tilde{y}\}$.
Consequently, $\delta_i(z)$ is a complete return word to the $G_\mathcal{A}$-orbit of $x$ and thus $\delta_i(z)=w$, which implies that $u=z$ and that $u$ is a palindrome.

Now, the $G_\mathcal{A}$-orbit of any word $w=\delta_i(u)$ with $u$ palindrome has two elements.
Indeed, either $w$ is even and $w^{-1}=\tau(w)$, or $w$ is odd and $w^{-1}=w$.
Thus such a $w$ is fixed by a nontrivial element of $G_\mathcal{A}$.
\end{proof}

\begin{example}
\label{ex:fibodoubleh2}
Let $S$ be the specular set of Example~\ref{ex:fibodouble2}.
Since it is a doubling of the Fibonacci set (which is Sturmian and thus full), it is $G_\mathcal{A}$-full with respect to the group $G_\mathcal{A}$ generated by the antimorphism $\sigma$ and the morphism $\tau$ of Example~\ref{ex:h2}.
The $G_\mathcal{A}$-orbit of $x=a$ is the set $X=\{a,c\}$.
The set of complete return words to $X$ (see also Section~\ref{sec:return}) is given by
\begin{displaymath}
\CR_S(X)=\{ aa, aba, cc, cdc\}.
\end{displaymath}
The four words are palindromes and thus they are fixed by $\sigma\tau$.

As another example, consider $x=ab$.
Its $G_\mathcal{A}$-orbit is the set $X = \{ ab, ba, cd, dc \}$ and the set of complete return words to $X$ is given by
$$
\CR_S(X)=\{ aba,baab,bab,cdc,dccd,dcd\}.
$$
Each of them is a palindrome, thus is fixed by $\sigma \tau$.
\end{example}

\begin{example}
Let $S$ be the specular set of Example~\ref{exampleFiboDouble}.
Since it is  a doubling of the Fibonacci set (which is Sturmian and thus full), it is $G_\mathcal{A}$-full with respect to the group $G_\mathcal{A}$ generated by the map $\sigma$ taking the inverse (that is fixing $a,c$ and exchanging $b$ and $d$) and the morphism $\tau$ (which exchanges $a,c$ and $b,d$ respectively).
The $G_\mathcal{A}$-orbit of $x=a$ is the set $X=\{a,c\}$.
We have
\begin{displaymath}
\CR_S(X)=\{abc,ac,ca,cda\}.
\end{displaymath}
The four words are fixed by $\sigma\tau$.
 As another example, consider $x=ab$. Then
$X=\{ab,bc,cd,da\}$ and $\CR_S(X)=\{abc,bcad,bcd,cda,dab,dacb\}$.
Each of them is fixed by some nontrivial element of $G_\mathcal{A}$.
\end{example}

\section{Linear involutions}
\label{sectionInvolutions}
In this section we define linear involutions and connections.
We prove that the natural coding of a linear involution without connections is a specular set (Theorem~\ref{theoremInvolutionSpecular}).

\subsection{Definition}
Let $A$ be an alphabet of cardinality $k$ with an involution $\theta$ and the corresponding specular group $G_\theta$. 
Note that we allow
$\theta$ to have fixed points. 
This leads to a definition of linear involutions which is somewhat more general than the one used in~\cite{DanthonyNogueira1988,BertheDelecroixDolcePerrinReutenauerRindone2014}.

We consider two copies $I\times \{0\}$ and
$I\times \{1\}$ of an open interval $I$ of the real line and denote  $\hat{I}=I\times \{0,1\}$.
We call the sets $I\times \{0\}$ and $I\times \{1\}$ the two
\emph{components} of $\hat{I}$. We consider each component as an open interval.

A \emph{generalized permutation} on $A$ of type $(\ell,m)$, with $\ell+m=k$,
  is a bijection $\pi:\{1,2,\ldots,k\}\rightarrow A$.
We represent it by a two line array
\begin{displaymath}
\pi=\begin{pmatrix} \pi(1)\ \pi(2)\ \ldots \pi(\ell)\\
\pi(\ell +1)\ \ldots \pi(\ell+m)
\end{pmatrix}
\end{displaymath}
A \emph{length data} associated with $(\ell,m,\pi)$ is a nonnegative
vector $\lambda\in \R_+^{A}=\R_+^{k}$ such that
\begin{displaymath}
\lambda_{\pi(1)}+\ldots+\lambda_{\pi(\ell)}=
\lambda_{\pi(\ell+1)}+\ldots+\lambda_{\pi(k)}
\text{ and }\lambda_a=\lambda_{a^{-1}}\text{ for all }a\in A.
\end{displaymath}

We consider a partition of $I\times \{0\}$ (minus $\ell-1$
points) in $\ell$ open intervals
$I_{\pi(1)},\ldots,I_{\pi(\ell)}$ of lengths $\lambda_{\pi(1)},\ldots,\lambda_{\pi(\ell)}$
and a partition of $I\times \{1\}$ (minus $m-1$ points) in $m$ open intervals
$I_{\pi(\ell+1)},\ldots,I_{\pi(\ell+m)}$ of lengths $\lambda_{\pi(\ell+1)},\ldots,\lambda_{\pi(\ell+m)}$. Let $\Sigma$ be the set of $k-2$ \emph{division points} separating
the intervals $I_a$ for $a\in A$.

The \emph{linear involution} on $I$ relative to these data is the
map $T=\sigma_2\circ\sigma_1$ defined on the set
$\hat{I}\setminus\Sigma$ as 
the composition
of two involutions defined as follows. 
\begin{enumerate}
\item[(i)]The first involution $\sigma_1$ is defined on $\hat{I}\setminus\Sigma$.
It is such that for each $a\in A$, its restriction to $I_a$
is either a translation or a symmetry from $I_a$ onto $I_{a^{-1}}$.

\item[(ii)]The second involution exchanges the two components of
  $\hat{I}$.
It  is defined for $(x,\delta)\in \hat{I}$
by $\sigma_2(x,\delta)=(x,1-\delta)$. The image of $z$ by $\sigma_2$
is called the \emph{mirror image} of $z$.
\end{enumerate}
We also say that $T$ is a linear involution on $I$ and relative to
 the alphabet $A$
or that it is a $k$-linear involution to express the fact
that the alphabet $A$ has $k$ elements.

\begin{example}\label{exampleLinear}
Let $A=\{a,b,c,d,a^{-1},b^{-1},c^{-1},d^{-1}\}$ and
\begin{displaymath}
\pi=\begin{pmatrix}a&b&a^{-1}&c\\c^{-1}&d^{-1}&b^{-1}&d
\end{pmatrix}.
\end{displaymath}
Let $T$ be the $8$-linear involution corresponding to the length data
represented in Figure~\ref{figureLinear} (we represent $I\times\{0\}$
above $I\times \{1\}$) with the assumption that the restriction
of $\sigma_1$ to $I_a$ and $I_d$ is a symmetry while its restriction
to $I_b,I_c$ is a translation.

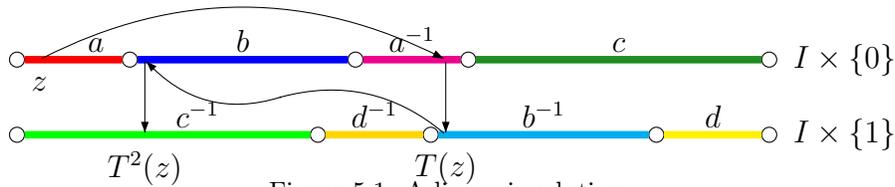
\begin{figure}[hbt]
\centering
\gasset{Nadjust=wh,AHnb=0}
\begin{picture}(115,17)(0,-1)
\node(h0)(0,10){}\node(h1)(15,10){}\node(h2)(45,10){}\node(h3)(60,10){}\node(h4)(100,10){}\node[Nframe=n](0)(110,10){$I\times\{0\}$}
\node(b0)(0,0){}\node(b1)(40,0){}\node(b2)(55,0){}\node(b3)(85,0){}\node(b4)(100,0){}\node[Nframe=n](1)(110,0){$I\times\{1\}$}
\gasset{Nh=.1,Nw=.1,Nadjust=n}
\node[ExtNL=y,Nh=.6,Nw=.6,Nfill=y,NLangle=-90,NLdist=2](z)(3,10){$z$}\node[Nh=.6,Nw=.6,Nfill=y](sz)(57,10){}\node[ExtNL=y,Nh=.6,Nw=.6,Nfill=y,NLangle=-90,NLdist=2](Tz)(57,0){$T(z)$}
\node[Nh=.6,Nw=.6,Nfill=y](sTz)(17,10){}\node[ExtNL=y,Nh=.6,Nw=.6,Nfill=y,NLangle=-90,NLdist=2](TTz)(17,0){$T^2(z)$}
\node(m)(35,5){}


\drawedge[linecolor=red,linewidth=1,ELpos=70](h0,h1){$a$}
\drawedge[linecolor=blue,linewidth=1](h1,h2){$b$}
\drawedge[linecolor=magenta,linewidth=1](h2,h3){$a^{-1}$}
\drawedge[linecolor=forestgreen,linewidth=1](h3,h4){$c$}
\drawedge[linecolor=green,linewidth=1,ELpos=60](b0,b1){$c^{-1}$}
\drawedge[linecolor=golden,linewidth=1](b1,b2){$d^{-1}$}
\drawedge[linecolor=cyan,linewidth=1](b2,b3){$b^{-1}$}
\drawedge[linecolor=yellow,linewidth=1](b3,b4){$d$}
\gasset{AHnb=1}
\drawedge[AHnb=1,curvedepth=7](z,sz){}\drawedge[AHnb=1](sz,Tz){}
\drawedge[curvedepth=-3,AHnb=0](Tz,m){}\drawedge[curvedepth=3](m,sTz){}\drawedge(sTz,TTz){}
\end{picture}
\caption{A linear involution.}\label{figureLinear}
\end{figure}
We indicate on the figure the effect of the transformation $T$ on a point
$z$ located in the left part of the interval $I_a$. The point
$\sigma_1(z)$ is located in the right part of $I_{a^{-1}}$ and the point
$T(z)=\sigma_2\sigma_1(z)$ is just below on the left of $I_{b^{-1}}$.
 Next, the point $\sigma_1T(z)$ is located on the left part of $I_b$
and the point $T^2(z)$ just below.
\end{example}
Thus the notion of linear involution is an extension of the notion 
of  interval exchange transformation in the following sense.
Assume that 
\begin{enumerate}
\item[(i)] $\ell=m$, 
\item[(ii)] for each letter $a\in A$, the interval
$I_a$ belongs to $I\times\{0\}$ if and only if $I_{a^{-1}}$ belongs
to $I\times\{1\}$,  
\item[(iii)] the restriction of $\sigma_1$
to each subinterval is a translation. 
\end{enumerate}
Then, the restriction of $T$
to $I\times \{0\}$ is an interval exchange (and so is its restriction to 
$I\times \{1\}$ which is the inverse of the first one). Thus,
in this case, $T$ is a pair of mutually inverse interval exchange transformations.

It is also an extension of the notion of interval exchange with flip
\cite{Nogueira1989,NogueiraPiresTroubetzkoy2013}. Assume again  conditions (i)
and (ii), but now that the restriction of $\sigma_1$
to at least one  subinterval is a symmetry. Then the restriction of $T$
to $I\times \{0\}$ is an interval exchange  with flip.

Note that for convenience we consider in this paper interval exchange transformations
defined by a partition of an open interval  minus
$\ell-1$ points in $\ell$
 open intervals. The usual notion of interval exchange transformation
uses a partition of a semi-interval in a finite number of semi-intervals.

A linear involution $T$ is a bijection from $\hat{I}\setminus\Sigma$
onto $\hat{I}\setminus\sigma_2(\Sigma)$.
Since $\sigma_1,\sigma_2$ are involutions and $T=\sigma_2\circ\sigma_1$,
 the inverse of $T$
is $T^{-1}=\sigma_1\circ\sigma_2$.

The set $\Sigma$ of division points is also the set of singular points
of $T$ and their mirror images are the singular points of $T^{-1}$
(which are the points where $T$ (resp. $T^{-1}$) is not defined).
Note that these singular points $z$ may be `false' singularities, in the sense
that $T$ can have a continuous extension to an open neighborhood of $z$.

Two particular cases of linear involutions deserve attention.

A linear involution $T$ on the alphabet $A$ 
relative to a generalized permutation $\pi$ of type
$(\ell,m)$ 
is said to be \emph{nonorientable} if there are indices $i,j\le \ell$ such that
$\pi(i)=\pi(j)^{-1}$  (and thus indices $i,j\ge \ell+1$
such that $\pi(i)=\pi(j)^{-1}$).  In other words,  there is  some $a\in A$ for which  $I_a$
and $I_{a^{-1}}$ belong to  the same component  of $\hat{I}$. Otherwise
$T$ is said to be \emph{orientable}.

A linear involution $T=\sigma_2\circ \sigma_1$ on $I$
relative to the alphabet $A$
is said to be \emph{coherent} if, for each $a\in A$, the restriction
of $\sigma_1$ to $I_a$ is a translation if and only if $I_a$
and $I_{a^{-1}}$ belong to distinct components of $\hat{I}$. 

\begin{example}
The linear involution of Example~\ref{exampleLinear} is coherent.
\end{example}

Linear involutions which are orientable and coherent
correspond to interval exchange transformations, 
whereas  orientable but noncoherent  linear
involutions are interval exchanges with flip.

Orientable linear involutions correspond to orientable laminations (see
\cite{BertheDelecroixDolcePerrinReutenauerRindone2014}), whereas 
 coherent
linear involutions correspond to orientable surfaces. Thus coherent
nonorientable involutions correspond to nonorientable laminations
on orientable surfaces.

\subsection{Minimal involutions}

A \emph{connection} of a linear involution $T$ is a triple $(x,y,n)$ 
where $x$ is a singularity of $T^{-1}$, $y$ is a singularity of $T$,
$n\ge 0$ and $T^n x = y$. 

\begin{example}\label{exampleConnection}
Let us consider the linear involution $T$ which is the same as in Example~\ref{exampleLinear} but such that the restriction of $\sigma_1$ to $I_c$ 
is a symmetry. Thus $T$ is not coherent. We assume that $I=]0,1[$,
that $\lambda_a=\lambda _d$. Let $x=(1-\lambda_d,0)$ and $y=(\lambda_a,0)$.

Then $x$ is a singularity of $T^{-1}$ ($\sigma_2(x)$ is the left endpoint
of $I_d$), $y$ is a singularity of $T$ (it is the right endpoint of $I_a$)
and $T(x)=y$. Thus $(x,y,1)$ is a connection.
\end{example}

\begin{example}\label{exampleInvolution3}
Let $T$ be the  linear involution on $I=]0,1[$
represented in Figure~\ref{figureLinear3}. We assume that the
restriction of $\sigma_1$ to $I_a$ is a translation
whereas the restriction to $I_b$ and $I_c$ is a symmetry.
We choose $(3-\sqrt{5})/2$ for the length of the interval $I_c$
(or $I_b$). With this choice, $T$ has no connections.

\begin{figure}[hbt]
\centering
\gasset{AHnb=0,Nadjust=wh}
\begin{picture}(100,15)
\node(h0)(0,10){}\node(b)(23.6,10){}\node(bbar)(61.8,10){}\node(h1)(100,10){}
\node(b0)(0,0){}\node(cbar)(38.2,0){}\node(abar)(76.4,0){}\node(b1)(100,0){}

\drawedge[linecolor=red,linewidth=1](h0,b){$a$}
\drawedge[linecolor=blue,linewidth=1](b,bbar){$b$}
\drawedge[linecolor=cyan,linewidth=1](bbar,h1){$b^{-1}$}
\drawedge[linecolor=forestgreen,linewidth=1](b0,cbar){$c$}
\drawedge[linecolor=green,linewidth=1](cbar,abar){$c^{-1}$}
\drawedge[linecolor=magenta,linewidth=1](abar,b1){$a^{-1}$}
\end{picture}
\caption{A linear involution without connections.}\label{figureLinear3}
\end{figure}
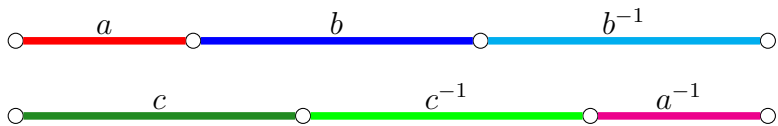
\end{example}

Let $T$ be a  linear involution without connections. Let 
\begin{equation}
O=\bigcup_{n\ge 0}T^{-n}(\Sigma) \quad \text{and }\quad \hat{O}=O\cup \sigma_2(O)\label{eqO}
\end{equation}
be respectively the negative orbit of the singular points and its
closure under mirror image.
Then $T$ is a bijection from $\hat{I}\setminus \hat{O}$ onto itself.
Indeed, assume that $T(z)\in \hat{O}$. If $T(z)\in O$ then
$z\in O$. Next if $T(z)\in \sigma_2(O)$, then
$T(z)\in \sigma_2(T^{-n}(\Sigma))=T^n(\sigma_2(\Sigma))$ for some $n\ge 0$. We cannot
have $n=0$ since $\sigma_2(\Sigma)$ is not in the image of $T$.
Thus $z\in T^{n-1}(\sigma_2(\Sigma))=\sigma_2(T^{-n+1}(\Sigma))\subset
\sigma_2(O)$. Therefore in both cases $z\in \hat{O}$. The converse
implication
is proved in the same way.

 A linear involution $T$ 
on $I$ without connections
is  minimal if for any point $z\in\hat{I}\setminus\hat{O}$ 
the nonnegative orbit of $z$  is dense in $\hat{I}$. 

Note that when a linear involution is  orientable, that is,  when it is a pair
of interval exchange transformations (with or without flips),   the interval exchange transformations
can be minimal although the linear involution is not since each component
of $\hat{I}$ is stable by the action of $T$. 
Moreover, it is shown in~\cite{DanthonyNogueira1990} that noncoherent linear 
involutions are almost surely not minimal.

\begin{example}\label{exampleNonCoherent}
Let us consider the noncoherent linear involution $T$ which is the same as in Example~\ref{exampleLinear} but such that the restriction of $\sigma_1$ to $I_c$ 
is a symmetry, as in Example~\ref{exampleConnection}. We assume that $I=]0,1[$,
that $\lambda_a=\lambda _d$ and that $1/4<\lambda_c<1/2$
 and that $\lambda_a+\lambda_b<1/2$.
 Let $x=1/2+\lambda_c$ and $z = (x,0)$ (see Figure~\ref{figureLinearnonCoherent}).
We have then $T^3(z)=z$, showing that $T$ is not minimal. Indeed,
since $z\in I_c$, we have
$T(z)=(1-x,0)=(1/2-\lambda_c,0)$. Since
$T(z)\in I_a$ we have $T^2(z)=((\lambda_a+\lambda_b)+(\lambda_a-1+x),1)=(x-\lambda_c,1)=(1/2,1)$.
Finally, since $T^2(z)\in I_{d^{-1}}$, we obtain $(1,0)-T^3(z)=T^2(z)-(\lambda_c,1)=(1,0)-z$
and thus $T^3(z)=z$.
\begin{figure}[hbt]
\centering
\gasset{Nadjust=wh,AHnb=0}
\begin{picture}(100,18)(0,-2)
\node(h0)(0,10){}\node(h1)(15,10){}\node(h2)(45,10){}\node(h3)(60,10){}\node(h4)(100,10){}
\node(b0)(0,0){}\node(b1)(40,0){}\node(b2)(55,0){}\node(b3)(85,0){}\node(b4)(100,0){}
\gasset{Nh=.6,Nw=.6,Nfill=y,Nadjust=n,ExtNL=y}
\node[NLdist=2](z)(90,10){$z=T^3(z)$}\node[Nh=.1,Nw=.1](m)(50,5){}\node(sz)(10,0){}\node[NLangle=-140,NLdist=2](Tz)(10,10){$T(z)$}\node(sTz)(50,10){}\node[ExtNL=y,NLangle=-90,NLdist=2](TTz)(50,0){$T^2(z)$}\node(sTTz)(90,0){}
\drawedge[linecolor=red,linewidth=1,ELpos=70](h0,h1){$a$}
\drawedge[linecolor=blue,linewidth=1](h1,h2){$b$}
\drawedge[linecolor=magenta,linewidth=1](h2,h3){$a^{-1}$}
\drawedge[linecolor=forestgreen,linewidth=1](h3,h4){$c$}
\drawedge[linecolor=green,linewidth=1,ELpos=60](b0,b1){$c^{-1}$}
\drawedge[linecolor=yellow,linewidth=1](b1,b2){$d^{-1}$}
\drawedge[linecolor=cyan,linewidth=1](b2,b3){$b^{-1}$}
\drawedge[linecolor=golden,linewidth=1](b3,b4){$d$}
\gasset{AHnb=1}
\drawedge[curvedepth=3,AHnb=0](z,m){}\drawedge[curvedepth=-3](m,sz){}
\drawedge(sz,Tz){}\drawedge[curvedepth=4](Tz,sTz){}\drawedge(sTz,TTz){}
\drawedge[curvedepth=-4](TTz,sTTz){}\drawedge(sTTz,z){}
\end{picture}
\caption{A noncoherent linear involution.}\label{figureLinearnonCoherent}
\end{figure}
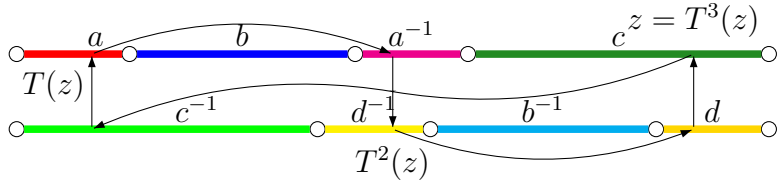

\end{example}

The following result (already  proved in~\cite[Proposition 4.2]{BoissyLanneau2009}
 for the class of coherent involutions)
is~\cite[Proposition 3.7]{BertheDelecroixDolcePerrinReutenauerRindone2014}.
The proof uses Keane's theorem proving that an
interval exchange transformation without connections is minimal~\cite{Keane1975}. 
\begin{proposition}\label{propBL}
Let $T$ be a linear involution without connections on $I$. If $T$
is nonorientable, it is minimal. Otherwise, its restriction to each component
of $\hat{I}$ is minimal. 
\end{proposition}

\subsection{Natural coding}
Let $T$ be a linear involution on $I$, let $\hat{I}=I\times\{0,1\}$ and let
$\hat{O}$ be the set defined by Equation~\eqref{eqO}.

Given $z\in \hat{I}\setminus \hat{O}$, the \emph{infinite natural
  coding} of $T$ relative to $z$ is the infinite word
$\Sigma_T(z)=a_0a_1\ldots$
on the alphabet $A$ defined by
\begin{displaymath}
a_n=a\quad\text{ if }\quad T^n(z)\in I_a.
\end{displaymath}
We first observe that the infinite word $\Sigma_T(z)$ is
reduced. Indeed, assume that $a_n=a$ and $a_{n+1}=a^{-1}$ with $a\in
A$. Set $x=T^n(z)$ and $y=T(x)=T^{n+1}(z)$. Then $x\in I_a$
and $y\in I_{a^{-1}}$. But $y=\sigma_2(u)$ with $u=\sigma_1(x)$. Since $x\in
I_a$, we have $u\in I_{a^{-1}}$. This implies that
$y=\sigma_2(u)$ and $u$ belong to the same component of $\hat{I}$, a contradiction.

We denote by $\LL(T)$ the set of factors of the infinite natural codings
of $T$. We say that $\LL(T)$ is 
the \emph{natural coding} of $T$. 

\begin{example}
Let $T$ be the linear involution of Example~\ref{exampleInvolution3}.
The words of length at most $3$ of $S=\LL(T)$ are represented 
in Figure~\ref{figureSetS}.
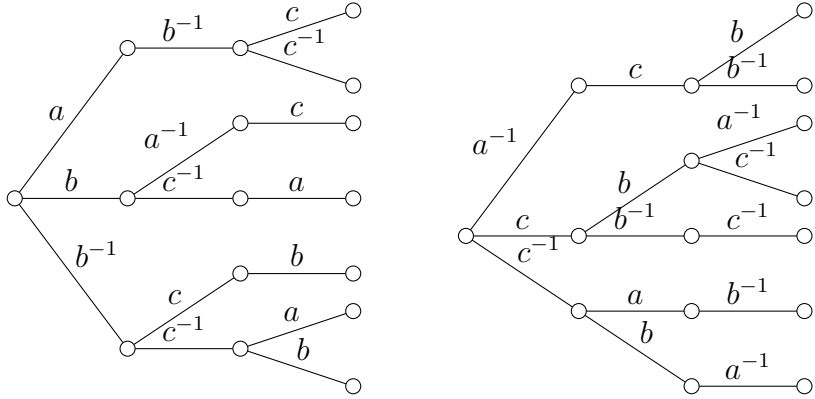
\begin{figure}[hbt]
\centering\gasset{Nadjust=wh,AHnb=0}
\begin{picture}(120,50)
\put(0,0){
\begin{picture}(60,40)(0,-5)
\node(1)(0,20){}
\node(a)(15,40){}
\node(b)(15,20){}\node(bbar)(15,0){}
\node(abbar)(30,40){}
\node(babar)(30,30){}\node(bcbar)(30,20){}
\node(bbarc)(30,10){}\node(bbarcbar)(30,0){}
\node(abbarc)(45,45){}\node(abbarcbar)(45,35){}
\node(babarc)(45,30){}\node(bcbara)(45,20){}
\node(bbarcb)(45,10){}
\node(bbarcbara)(45,5){}\node(bbarcbarb)(45,-5){}

\drawedge(1,a){$a$}\drawedge(1,b){$b$}\drawedge(1,bbar){$b^{-1}$}
\drawedge(a,abbar){$b^{-1}$}
\drawedge(b,bcbar){$c^{-1}$}\drawedge(b,babar){$a^{-1}$}
\drawedge(bbar,bbarc){$c$}\drawedge(bbar,bbarcbar){$c^{-1}$}
\drawedge(abbar,abbarc){$c$}\drawedge(abbar,abbarcbar){$c^{-1}$}
\drawedge(babar,babarc){$c$}\drawedge(bcbar,bcbara){$a$}
\drawedge(bbarc,bbarcb){$b$}
\drawedge(bbarcbar,bbarcbara){$a$}\drawedge(bbarcbar,bbarcbarb){$b$}
\end{picture}
}
\put(60,0){
\begin{picture}(60,40)
\node(1)(0,20){}
\node(abar)(15,40){}\node(c)(15,20){}\node(cbar)(15,10){}
\node(abarc)(30,40){}
\node(cb)(30,30){}\node(cbbar)(30,20){}
\node(cbara)(30,10){}\node(cbarb)(30,0){}
\node(abarcb)(45,50){}\node(abarcbbar)(45,40){}
\node(cbabar)(45,35){}\node(cbcbar)(45,25){}
\node(cbbarcbar)(45,20){}
\node(cbarabbar)(45,10){}\node(cbarbabar)(45,0){}

\drawedge(1,abar){$a^{-1}$}\drawedge(1,c){$c$}\drawedge(1,cbar){$c^{-1}$}
\drawedge(abar,abarc){$c$}
\drawedge(c,cb){$b$}\drawedge(c,cbbar){$b^{-1}$}
\drawedge(cbar,cbara){$a$}\drawedge(cbar,cbarb){$b$}
\drawedge(abarc,abarcb){$b$}\drawedge(abarc,abarcbbar){$b^{-1}$}
\drawedge(cb,cbabar){$a^{-1}$}\drawedge(cb,cbcbar){$c^{-1}$}
\drawedge(cbbar,cbbarcbar){$c^{-1}$}
\drawedge(cbara,cbarabbar){$b^{-1}$}\drawedge(cbarb,cbarbabar){$a^{-1}$}
\end{picture}
}
\end{picture}
\caption{The words of length at most $3$ of $S$.}\label{figureSetS}
\end{figure}

The set $S$ can  actually be defined directly as the set of factors
of the substitution
\begin{displaymath}
f:a\mapsto cb^{-1},\quad b\mapsto c,\quad c\mapsto ab^{-1}.
\end{displaymath}
which extends to an automorphism of the free group on $\{a,b,c\}$
(see~\cite{BertheDelecroixDolcePerrinReutenauerRindone2014}).

\end{example}

The following is Proposition 5.3 in~\cite{BertheDelecroixDolcePerrinReutenauerRindone2014}.
\begin{proposition} \label{prop:inverse}
The natural coding of a linear involution is closed under taking inverses.
\end{proposition}

We prove the following result.
\begin{theorem}\label{theoremInvolutionSpecular}
The natural coding of a linear involution without connections
 is a specular set.
\end{theorem}
\begin{proof}
Let $T$ be a linear involution without connections. By Proposition~\ref{prop:inverse}, the set $\LL(T)$ is symmetric. Since it is by definition biextendable
and formed of reduced words,
it is a laminary set. 
By~\cite[Theorem 9.5]{DolcePerrin2016}, $\LL(T)$ is a tree set of characteristic $2$. Thus $\LL(T)$ is specular.
\end{proof}

We now present an example of a linear involution on an alphabet $A$ where the involution $\theta$ has fixed points.

\begin{example}
\label{exampleFiboDoubleInvolution}
Let $A=\{a,b,c,d\}$ be as in Example~\ref{exampleSpecularGroup} (in particular, $d=b^{-1}$, $a=a^{-1}$, $c=c^{-1}$).
\begin{figure}[hbt]
\centering
\gasset{AHnb=0,Nadjust=wh}
\begin{picture}(100,15)
\node(h0)(0,10){}\node(bbar)(61.8,10){}\node(h1)(100,10){}
\node(b0)(0,0){}\node(c)(38.2,0){}\node(b1)(100,0){}

\drawedge[linecolor=red,linewidth=1](h0,bbar){$a$}
\drawedge[linecolor=yellow,linewidth=1](bbar,h1){$d$}
\drawedge[linecolor=blue,linewidth=1](b0,c){$b$}
\drawedge[linecolor=green,linewidth=1](c,b1){$c$}
\end{picture}
\caption{A linear involution on $A=\{a,b,c,d\}$.}\label{figureLinearFiboDouble}
\end{figure}
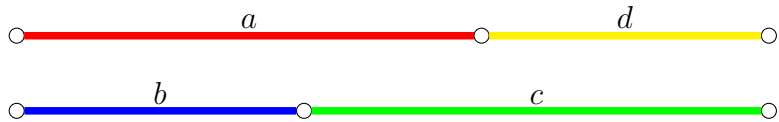
Let $T$ be the linear involution represented in Figure~\ref{figureLinearFiboDouble} with $\sigma_1$ being a translation on $I_b$ and a symmetry on $I_a,I_c$.
Choosing $(3-\sqrt{5})/2$ for the length of $I_b$, the involution is without connections. 
Thus $S = \LL(T)$ is a specular set.
Let us show it is equal to the specular set obtained by the doubling transducer in Example~\ref{exampleFiboDouble}.
Indeed, consider the interval exchange $V$ on the interval $Y=]0,2[$ represented in Figure~\ref{figureFibonacciDouble} on the right, which is obtained by using two copies of the  interval exchange $U$ defining the Fibonacci set (represented in Figure~\ref{figureFibonacciDouble} on the left).

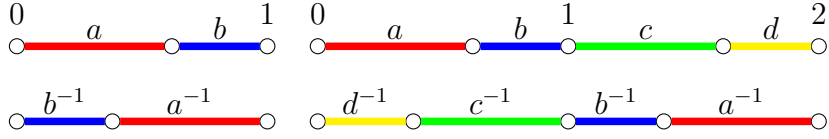
\begin{figure}[hbt]
\centering\gasset{AHnb=0,Nh=2,Nw=2,ExtNL=y,NLdist=2}
\begin{picture}(100,15)
\put(0,0){
\begin{picture}(40,15)
\node(h0)(0,10){$0$}\node(b)(20.6,10){}\node(h1)(33.3,10){$1$}
\node(b0)(0,0){}\node(a)(12.7,0){}\node(b1)(33.3,0){}

\drawedge[linecolor=red,linewidth=1](h0,b){$a$}\drawedge[linecolor=blue,linewidth=1](b,h1){$b$}
\drawedge[linecolor=blue,linewidth=1](b0,a){$b^{-1}$}\drawedge[linecolor=red,linewidth=1](a,b1){$a^{-1}$}
\end{picture}
}
\put(40,0){
\begin{picture}(40,15)
\node(h0)(0,10){$0$}\node(b)(20.6,10){}\node(h1)(33.3,10){$1$}
\node(d)(53.9,10){}\node(h2)(66.6,10){$2$}

\node(b0)(0,0){}\node(cbar)(12.7,0){}\node(b1)(33.3,0){}
\node(abar)(46,0){}\node(b2)(66.6,0){}

\drawedge[linecolor=red,linewidth=1](h0,b){$a$}\drawedge[linecolor=blue,linewidth=1](b,h1){$b$}
\drawedge[linecolor=green,linewidth=1](h1,d){$c$}
\drawedge[linecolor=yellow,linewidth=1](d,h2){$d$}

\drawedge[linecolor=yellow,linewidth=1](b0,a){$d^{-1}$}
\drawedge[linecolor=green,linewidth=1](a,b1){$c^{-1}$}
\drawedge[linecolor=blue,linewidth=1](b1,abar){$b^{-1}$}
\drawedge[linecolor=red,linewidth=1](abar,b2){$a^{-1}$}
\end{picture}
}
\end{picture}
\caption{Interval exchanges $U$ and $V$ for the Fibonacci set and its doubling.}
\label{figureFibonacciDouble}
\end{figure}

Let $X=]0,1[\times\{0,1\}$ and let $\alpha:Y\rightarrow X$ be the
map defined by
\begin{displaymath}
\alpha(z)=\begin{cases}(z,0)&\text{if $z\in]0,1[$}\\(2-z,1)&\text{otherwise.}
\end{cases}
\end{displaymath}
Then $\alpha\circ V=T\circ\alpha$ and thus $\LL(V)=\LL(T)$.
The interval exchange $V$ is actually the orientation covering
of the linear involution $T$ (see~\cite{BertheDelecroixDolcePerrinReutenauerRindone2014}).
\end{example}

\section{Return words}
\label{sec:return}
In this section we introduce three variants of the notion of return words, namely complete, right and mixed return words.
We prove several results concerning sets of return words (Theorems~\ref{theoremCardRightReturns}, \ref{theoremCardReturns}, \ref{theoremCardFirstMixed}).
We also prove that the set of return words to a given word forms a basis of the even subgroup (Theorem~\ref{theoremReturns} referred to as the First Return Theorem) and that the mixed return words form a monoidal basis of the specular group (Theorem \ref{theoremFirstMixed}).

\subsection{Cardinality Theorems for return words}
\label{subsec:cardret}
In this section, we introduce several notions of return words: complete return words, right (or left) return words and  mixed return words.
For each of them, we prove a cardinality theorem (Theorems~\ref{theoremCardRightReturns}, \ref{theoremCardReturns} and \ref{theoremCardFirstMixed}).

Here, when we consider a recurrent set $S$ containing the alphabet $A$, we implicitly assume that all words of $S$ are on the alphabet $A$.

\subsubsection{Complete return words}
Let $S$ be a factorial set of words and let $X \subset S$ be a set of nonempty words.
A \emph{complete return word} to $X$ is a word of $S$ with a proper prefix in $X$, a proper suffix in $X$ but no internal factor in $X$.
We denote by $\CR_S(X)$ the set of complete return words to $X$.

The set $\CR_S(X)$ is a bifix code.
If $S$ is uniformly recurrent, $\CR_S(X)$ is finite for any finite set $X$.
For $x\in S$, we denote $\CR_S(x)$ instead of $\CR_S(\{x\})$.
Thus $\CR_S(x)$ is the usual notion of a complete return word (see~\cite{Durand1998} for example).

\begin{example}
\label{exampleCompleteFiboDouble}
Let $S$ be the specular set of Example~\ref{exampleFiboDouble}.
One has
\begin{eqnarray*}
\CR_S(a)&=&\{abca,abcda,acda\}\\
\CR_S(b)&=&\{bcab,bcdacdab,bcdacdacdab\}\\
\CR_S(c)&=&\{cabc,cdabc,cdac\}\\
\CR_S(d)&=&\{dabcabcabcd,dabcabcd,dacd\}.
\end{eqnarray*}
\end{example}

The following result is proved in~\cite[Theorem 5.2]{DolcePerrin2016}.

\begin{theorem}
\label{propositionNew}
Let $S$ be a recurrent neutral set containing the alphabet $A$.
For any finite nonempty bifix code $X \subset S$ with empty kernel, we have
\begin{equation}
\Card(\CR_S(X)) = \Card(X)+\Card(A)-\chi(S).
\label{formulaComplete}
\end{equation}
\end{theorem}

As a consequence of Theorem~\ref{propositionNew}, one has the following statement. 

\begin{corollary}
\label{corollaryCardCompleteReturns}
Let $S$ be a recurrent specular set on the alphabet $A$.
For any finite nonempty bifix code $X\subset S$ with empty kernel, one has
\begin{displaymath}
\Card(\CR_S(X)) = \Card(X)+\Card(A)-2.
\end{displaymath}
\end{corollary}

The following example illustrates Corollary~\ref{corollaryCardCompleteReturns}.

\begin{example}
\label{exampleJulien3}
Let $S$ be the specular set on the alphabet $A=\{a,b,c,d\}$
of Example~\ref{exampleJulien}. 
We have
\begin{displaymath}
\CR_S(\{a,b\})=\{ab,acda,bca,bcda\}.
\end{displaymath}
It has four elements in agreement with Corollary~\ref{corollaryCardCompleteReturns}.
\end{example}

We note that when $X$ is a finite $S$-maximal bifix code
of $S$-degree $d$
with kernel $K(X)$, the
set $\CR_S(X)$ has the following property.
 For any set $K$ such that $K(X)\subset K\subset X$
with $K\ne X$, the set
$Y=K\cup \CR_S(X\setminus K)$
is an $S$-maximal bifix code of $S$-degree $d_S(X)+1$.
The code $X$ is the derived code of $Y$
(see~\cite[Section 4.3]{BerstelDeFelicePerrinReutenauerRindone2012}).
This gives a connection between Equations~\eqref{formulaBifixMax}
and \eqref{formulaComplete}. 
By Equation~\eqref{formulaBifixMax}, we have
\begin{displaymath}
\Card(Y)=(d+1)(\Card(A)-\chi(S))+\chi(S)=\Card(X)+\Card(A)-\chi(S).
\end{displaymath}
Thus 
\begin{eqnarray*}
\Card(\CR_S(X\setminus K))&=&\Card(Y)-\Card(K)\\
&=&\Card(X)-\Card(K)+\Card(A)-\chi(S)\\
&=&\Card(X\setminus K)+\Card(A)-\chi(S)\\
\end{eqnarray*}
which is Formula~\eqref{formulaComplete} since $X\setminus K$
is a bifix code with empty kernel.

\subsubsection{Right  return words}

Let $S$ be a factorial set. For any nonempty word 
$x\in S$, a \emph{right return
word} to $x$ in $S$ is a word $w$ such that $xw$ is a  complete return
word to $x$. One defines symmetrically the \emph{left  return words}
to $x\in S$
as the words $w$ such that $wx$ is a  complete return word.
We denote by $\RR_S(x)$ the set of right  return
words to $x$ in $S$ and by $\RR'_S(x)$ the corresponding set
of left  return words. 

Note that when $S$ is a laminary set $\RR_S(x)^{-1}=\RR'_S(x^{-1})$.

\begin{proposition}\label{propositionReturnsEven}
Let $S$ be a specular set and let $x\in S$ be a nonempty word.
All the words of $\RR_S(x)$ are even.
\end{proposition}
\begin{proof}
If $w\in\RR_S(x)$, we have $xw=vx$ for some $v\in S$. If $x$
is odd, assume that $x\in S_{0,1}$. Then $w\in S_{1,1}$.
Thus $w$ is even. If $x$ is even, assume that $x\in S_{0,0}$. Then
$w\in S_{0,0}$ and $w$ is even again.
\end{proof}

\begin{theorem}[Cardinality Theorem for right return words]
\label{theoremCardRightReturns}
Let $S$ be a recurrent specular set. For any $x\in S$,
the set $\RR_S(x)$ has $\Card(A)-1$ elements.
\end{theorem}
\begin{proof}
This follows directly from Corollary~\ref{corollaryCardCompleteReturns}
with $X=\{x\}$
since $\Card(\RR_S(x))=\Card(\CR_S(x))$.
\end{proof}

\begin{example}\label{exampleFiboDoubleFirst}
Let $S$ be the specular set of Example~\ref{exampleFiboDouble}. We have
\begin{eqnarray*}
\RR_S(a)&=&\{bca,bcda,cda\},\\
\RR_S(b)&=&\{cab,cdacdab,cdacdacdab\},\\
\RR_S(c)&=&\{abc,dabc,dac\},\\
\RR_S(d)&=&\{abcabcd,abcabcabcd,acd\}.
\end{eqnarray*}
\end{example}

It is shown in~\cite{BertheDeFeliceDolceLeroyPerrinReutenauerRindone2013a} that if $S$ is a (uniformly) recurrent tree set of characteristic $1$ containing the alphabet $B$, then for any $x\in S$,
one has $\Card(\RR_S(x))=\Card(B)$.
The relation with Theorem~\ref{theoremCardRightReturns} is as follows.
Let $X$ be the even code and let $X_0=X\cap S_{0,0} ,X_1=X\cap S_{1,1}$. Thus $X=X_0\cup X_1$.

One has $\Card(X_0)=\Card(A)-1$ by Theorem~\ref{corollaryCardinality}
(indeed, $\Card(X)=2\Card(A)-2$ and $\Card(X_0)=\Card(X_1)$).

Let $f$ be a coding morphism for $X$.
Then for any $x\in S_{0,0}$, the set $\RR_S(x)$ is in bijection, via the decoding
by $X_0$, with the set of  right return words to $f^{-1}(x)$. Since
$f^{-1}(S_{0,0})$ is a tree set on $B_0=f^{-1}(X_0)$, the set $\RR_S(x)$ has 
$\Card(A)-1$ elements, in agreement with Theorem~\ref{theoremCardRightReturns}.

\subsubsection{Mixed  return words}
Let $S$ be a laminary set.
For $w\in S$ such that $w\ne w^{-1}$, we consider complete return words to the set $X=\{w,w^{-1}\}$. 

\begin{example}
\label{exampleInvolution3bis}
Let $T$ be the  linear involution of Example~\ref{exampleInvolution3}.
We have 
\begin{eqnarray*}
\CR_S(\{a,a^{-1}\})&=&\{ab^{-1}cba^{-1},ab^{-1}cbc^{-1}a,a^{-1}cb^{-1}c^{-1}a,\\
&&\qquad  ab^{-1}c^{-1}ba^{-1},a^{-1}cbc^{-1}a,a^{-1}cb^{-1}c^{-1}ba^{-1}\}\\
\CR_S(\{b,b^{-1}\})&=&\{ba^{-1}cb,ba^{-1}cb^{-1},bc^{-1}ab^{-1},b^{-1}cb,
b^{-1}c^{-1}ab^{-1},b^{-1}c^{-1}b\},\\
\CR_S(\{c,c^{-1}\})&=&\{cba^{-1}c,cbc^{-1},cb^{-1}c^{-1},c^{-1}ab^{-1}c,
c^{-1}ab^{-1}c^{-1},c^{-1}ba^{-1}c\}.
\end{eqnarray*}
\end{example}

\begin{theorem}
\label{theoremCardReturns}
Let $S$ be a recurrent specular set containing the alphabet $A$.
For any $w\in S$ such that $w\ne w^{-1}$, the set of complete return words to $\{w,w^{-1}\}$ has $\Card(A)$ elements.
\end{theorem}
\begin{proof}
The statement results directly of Corollary~\ref{corollaryCardCompleteReturns}. 
\end{proof}

\begin{example}
\label{exampleCompleteFiboDoubleMixed}
Let $S$ be the specular set of Example~\ref{exampleFiboDouble}. In view of
the values of $\CR_S(b)$ and $\CR_S(d)$
given in Example~\ref{exampleCompleteFiboDouble},
we have
\begin{displaymath}
\CR_S(\{b,d\})=\{bcab,bcd,dab,dacd\}.
\end{displaymath}
\end{example}

Two words $u,v$
are said to \emph{overlap} if a nonempty suffix of one of them
is a prefix of the other. In particular a nonempty word overlaps with
itself.

We now consider  the return words to $\{w,w^{-1}\}$ with $w$
such that $w$ and $w^{-1}$ do not overlap. 
This is true for every $w$ in
a laminary set $S$ where the involution $\theta$
has no fixed point, in particular when $S$ is the natural coding of a linear
involution. In this case, the group $G_\theta$ is free and for any $w\in S$, the words $w$ and $w^{-1}$ do not overlap. 

With a complete return word $u$ to $\{w,w^{-1}\}$, we associate a word $N(u)$ obtained as follows. If $u$ has  $w$ as prefix, we erase it
and if $u$ has a suffix $w^{-1}$, we also erase it. Note that these two operations can be made in any order since
$w$ and $w^{-1}$ cannot overlap.

The \emph{mixed  return words to $w$} are the words $N(u)$ associated with complete return words $u$ to $\{w,w^{-1}\}$. We denote by $\MR_S(w)$ the set
of  mixed return words to $w$ in $S$.

Note that $\MR_S(w)$ is symmetric and that $w\MR_S(w)w^{-1}=\MR_S(w^{-1})$. Note
also that if $S$ is orientable, then
\begin{displaymath}
\MR_S(w)=\RR_S(w)\cup \RR_S(w)^{-1}=\RR_S(w)\cup \RR'_S(w^{-1}).
\end{displaymath}

The reason for this definition comes from the fact that, when $S$ is the natural coding of a linear involution, we are interested in the transformation induced on $I_w\cup\sigma_2(I_w)$,
where $I_w=I_{b_0}\cap T^{-1}(I_{b_1})\cap\ldots\cap T^{-m+1}(I_{b_{m-1}})$ for a word $w=b_0b_1\cdots b_{m-1}$ (see~\cite{BertheDelecroixDolcePerrinReutenauerRindone2014}).
The natural coding of a point in $I_w$ begins with $w$ while the natural coding of a point $z$ in $\sigma_2(I_w)$ `ends' with $w^{-1}$ in the sense that the natural coding of  $T^{-|w|}(z)$  begins with $w^{-1}$.

\begin{example}
\label{exampleMixedReturn}
Let $T$ be the linear involution of Example~\ref{exampleInvolution3}.
We have 
\begin{eqnarray*}
\MR_S(a)&=&\{b^{-1}cb,b^{-1}cbc^{-1}a,a^{-1}cb^{-1}c^{-1}a
,b^{-1}c^{-1}b,a^{-1}cbc^{-1}a,a^{-1}cb^{-1}c^{-1}b\}\\
\MR_S(b)&=&\{a^{-1}cb,a^{-1}c,c^{-1}a,b^{-1}cb,
b^{-1}c^{-1}a,b^{-1}c^{-1}b\},\\
\MR_S(c)&=&\{ba^{-1}c,b,b^{-1},c^{-1}ab^{-1}c,
c^{-1}ab^{-1},c^{-1}ba^{-1}c\}.
\end{eqnarray*}
\end{example}

Observe that any uniformly recurrent biinfinite word $x$ such that $F(x)=S$ can be uniquely written as a concatenation of  mixed return words (see Figure~\ref{figureFact}).
Note that successive occurrences of $w$ may overlap but that successive occurrences of $w$ and $w^{-1}$ cannot.

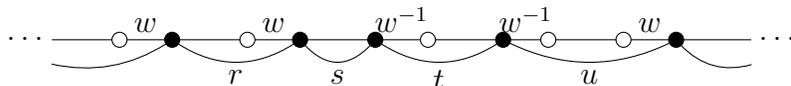
\begin{figure}[hbt]
\centering\gasset{Nadjust=wh,AHnb=0}
\begin{picture}(100,10)(0,-3)
\put(0,0){$\ldots$}
\node[Nframe=n](0)(5,0){}\node[Nframe=n](1b)(5,-3){}
\node(1)(15,0){}\node[fillcolor=black](2)(22,0){}
\node(3)(32,0){}\node[fillcolor=black](4)(39,0){}
\node[fillcolor=black](5)(49,0){}\node(6)(56,0){}\node[fillcolor=black](7)(66,0){}
\node(8)(72,0){}\node(9)(82,0){}\node[fillcolor=black](10)(89,0){}\node[Nframe=n](11)(100,0){}
\node[Nframe=n](11b)(100,-3){}

\drawedge(0,1){}\drawedge(1,2){$w$}\drawedge[curvedepth=-2](1b,2){}
\drawedge(2,3){}\drawedge(3,4){$w$}\drawedge[curvedepth=-3,ELside=r](2,4){$r$}
\drawedge(4,5){}\drawedge(5,6){$w^{-1}$}\drawedge[curvedepth=-3,ELside=r](5,7){$t$}
\drawedge[curvedepth=-3,ELside=r](4,5){$s$}\drawedge(6,7){}
\drawedge(7,8){$w^{-1}$}\drawedge(8,9){}\drawedge(9,10){$w$}
\drawedge[curvedepth=-3,ELside=r](7,10){$u$}\drawedge[curvedepth=-2](10,11b){}
\drawedge(10,11){}\put(100,0){$\ldots$}
\end{picture}
\caption{A uniformly recurrent infinite word factorized as an infinite product
$\cdots rstu\cdots$ of mixed   return words to $w$.}\label{figureFact}
\end{figure}

We have the following cardinality result. 

\begin{theorem}[Cardinality Theorem for mixed return words]
\label{theoremCardFirstMixed}
Let $S$ be a recurrent specular set on the alphabet $A$.
For any $w\in S$ such that $w,w^{-1}$ do not overlap, the set $\MR_S(w)$ has $\Card(A)$ elements.
\end{theorem}
\begin{proof}
This is a direct consequence of Theorem~\ref{theoremCardReturns} since
$\Card(\MR_S(w))=\Card(\CR_S(\{w,w^{-1}\})$ when $w$ and $w^{-1}$ do not overlap.
\end{proof}

Note that the bijection between $\CR_S(w,w^{-1})$ and $\MR_S(w)$ is illustrated in Figure~\ref{figureFact}.

\begin{example}
\label{exampleCardFirstMixed}
Let $S$ be the specular set of Example~\ref{exampleFiboDouble}. The
value of $\CR_S(b,d)$ is given in Example~\ref{exampleCompleteFiboDoubleMixed}.
Since $b,d$ do not overlap, 
\begin{displaymath}
\MR_S(b)=\{cab,c,dac,dab\}
\end{displaymath}
has four elements in agreement with Theorem~\ref{theoremCardFirstMixed}.
\end{example}
As a corollary, we obtain the following result.
\begin{corollary}
Let $S$ be the natural coding of a linear involution without
connections on the alphabet $A$. For any $w\in S$, the set
$\MR_S(w)$ has $\Card(A)$ elements.
\end{corollary}

\subsection{First Return Theorem}
By~\cite[Theorem 4.5]{BertheDeFeliceDolceLeroyPerrinReutenauerRindone2013a}, the set of right return words to a given word in a recurrent tree set of characteristic $1$ containing the alphabet $A$ is a basis of the free group on $A$.
We will see a counterpart of this result for recurrent specular sets.

Let $S$ be a specular set.
The \emph{even subgroup} is the group formed by the even words.
It is a subgroup of index $2$ of $G_\theta$ with symmetric rank $2(\Card(A)-1)$ by \eqref{SchreierSpecular} generated by the even code.
Since no even word is its own inverse (by Proposition~\ref{propositionPartition}), it is a free group.
Thus its rank is $\Card(A)-1$.

\begin{theorem}[First Return Theorem]
\label{theoremReturns}
Let $S$ be a recurrent specular set.
For any $w\in S$, the set of right return words to $w$ is a basis of the even subgroup.
\end{theorem}
\begin{proof}
We first consider the case where $w$ is even.
Let $f:B^*\rightarrow A^*$ be a coding morphism for the even code $X\subset S$.
Consider the partition $(S_{i,j})$, as in Proposition~\ref{propositionPartition}, and set $X_0=X\cap S_{0,0}$, $X_1=X\cap S_{1,1}$.
By Theorem~\ref{theoremDecodingEven}, the set $f^{-1}(S)$ is the union of the two recurrent tree sets of characteristic $1$, $T_0=f^{-1}(S_{0,0})$ and $T_1=f^{-1}(S_{1,1})$ on the alphabets $B_0=f^{-1}(X_0)$ and $B_1=f^{-1}(X_1)$ respectively.
We may assume that $w\in S_{0,0}$.
Then $\RR_S(w)$ is the image by $f$ of the set $R=\RR_{T_0}(f^{-1}(w))$. By \cite[Theorem 4.5]{BertheDeFeliceDolceLeroyPerrinReutenauerRindone2013a}, the set $R$ is a basis of the free group on $B_0$.
Thus $\RR_S(w)$ is a basis of the image of $F_{B_0}$ by $f$, which is the even subgroup.

Suppose now that $w$ is odd. Since the even code is an $S$-maximal bifix code, there exists  an odd word $u$ such that $uw\in S$.
Then $\RR_S(uw)\subset \RR_S(w)^*$.
By what precedes, the set $\RR_S(uw)$ generates the even subgroup and thus the group generated by $\RR_S(w)$ contains the even subgroup.
Since all words in $\RR_S(w)$ are even, the group generated by $\RR_S(w)$ is contained in the even subgroup, whence the equality.
We conclude by Theorem~\ref{theoremCardRightReturns}.
\end{proof}

\begin{example}
Let $S$ be the specular set of Example~\ref{exampleFiboDouble}.
The sets of right return words to $a,b,c,d$ are given in Example~\ref{exampleFiboDoubleFirst}.
Each one is a basis of the even subgroup.
\end{example}

Concerning mixed return words, we have the following statement.

\begin{theorem}
\label{theoremFirstMixed}
Let $S$ be a recurrent specular set.
For any $w\in S$ such that $w,w^{-1}$ do not overlap, the set $\MR_S(w)$ is a monoidal basis of the group $G_\theta$.
\end{theorem}
\begin{proof}
Since $w$ and $w^{-1}$ do not overlap, we have $\RR_S(w)\subset \MR_S(w)^*$.
Thus, by Theorem~\ref{theoremReturns}, the group $\langle\MR_S(w)\rangle$
contains the even subgroup. But $\MR_S(w)$ always contains odd
words. Indeed, assume that $w\in S_{i,j}$. Then $w^{-1}\in S_{1-j,1-i}$
and thus any $u\in\MR_S(w)$ such that $wuw^{-1}\in S$ is odd.
Since the even group is a maximal subgroup of $G_\theta$, 
this implies that $\MR_S(w)$ generates
the group $G_\theta$. Finally since $\MR_S(w)$ has $\Card(A)$ elements
by Theorem~\ref{theoremCardFirstMixed}, we obtain the conclusion
by Proposition~\ref{propSymBasis}.
\end{proof}
\begin{example}
Let $S$ be the specular set of Example~\ref{exampleFiboDouble}. We have
seen in Example~\ref{exampleCardFirstMixed} that
\begin{displaymath}
\MR_S(b)=\{c,cab,dab,dac\}.
\end{displaymath}
This set is a monoidal basis of $G_\theta$ in agreement with Theorem~\ref{theoremFirstMixed}.
\end{example}
Since, in the free group, a reduced word $w$ and its inverse do not overlap, we have the following corollary of Theorem~\ref{theoremFirstMixed}
in the case where the involution $\theta$ has no fixed points.
A geometric proof and interpretation  is
 given in~\cite{BertheDelecroixDolcePerrinReutenauerRindone2014}.
\begin{corollary}
Let $S$ be the natural coding of a linear involution without
connections  on the alphabet $A=B\cup B^{-1}$.
For any $w\in S$, the set
$\MR_S(w)$
is a monoidal basis of $F_B$.
\end{corollary}
\begin{example}
Let $T$ be the linear involution of Example~\ref{exampleInvolution3}.
We have seen in Example~\ref{exampleMixedReturn} that
$\MR_S(b)=\{a^{-1}cb,a^{-1}c,c^{-1}a,b^{-1}cb,b^{-1}c^{-1}a,b^{-1}c^{-1}b\}$.
It is a monoidal basis of the free group on $\{a,b,c\}$.
\end{example}

\section{Freeness and Saturation Theorems}
\label{sec:groups}
In this section we consider two notions concerning sets of generators of a subgroup $H$ in a specular group, namely free subsets and the set of prime words with respect to $H$.
We prove that a set closed by taking inverses is acyclic if and only if any symmetric bifix code is free (Theorem~\ref{theoremFreeness}).
Moreover, we prove that in such a set, for any finite symmetric bifix code $X$, the free monoid $X^*$ and the free subgroup $\langle X \rangle$ have the same intersection with $S$ (Theorem~\ref{theoremSaturation}).
To prove the last result we use the notion of coset automaton.

\subsection{Freeness Theorem}
Let $\theta$ be an involution on $A$ and let $G_\theta$ be the corresponding specular group.
A symmetric set $X$ is free if it is a monoidal basis of a subgroup $H$ of the  group $G_\theta$.
Thus a symmetric set $X\subset G_\theta$ is \emph{free} if for $x_1,x_2,\ldots,x_n\in X$, the product $x_1x_2\cdots x_n$ cannot reduce to $1$ unless $x_i = x_{i+1}^{-1}$ for some $i$ with $1 \leq i<n$.

The following is essentially Theorem 5.1 in~\cite{BertheDeFeliceDolceLeroyPerrinReutenauerRindone2013a}.

\begin{theorem}[Freeness Theorem]
\label{theoremFreeness}
A laminary set $S$ is acyclic if and only if any symmetric bifix code $X\subset S$ is free.
\end{theorem}
The proof is identical with that of Theorem 5.1 in \cite{BertheDeFeliceDolceLeroyPerrinReutenauerRindone2013a}, using the \emph{incidence graph} of a set $X$, which is the undirected graph $\G_X$ defined as follows.
Let $P$ be the set of proper prefixes of $X$ and let $Q$ be the set of its proper suffixes.
Set $P'=P\setminus\{\varepsilon\}$
and $Q'=Q\setminus \{\varepsilon\}$.
The set of vertices of $\G_X$ is the disjoint union of $P'$ and $Q'$.
The edges of $\G_X$ are the pairs $(p,q)$ for $p \in P'$ and $q \in Q'$ such that $pq \in X$.
As for the extension graph, we sometimes denote $1\otimes P',Q'\otimes 1$ the copies of $P',Q'$ used to define the set of vertices of $\G_X$.

\begin{example}
Let $S$ be a laminary set and let $X=S\cap A^2$ be the bifix code formed of the words of $S$ of length $2$.
The incidence graph of $X$ is identical with the extension graph $\E(\varepsilon)$.
\end{example}

The following statement is proved in \cite[Proposition 5.6]{BertheDeFeliceDolceLeroyPerrinReutenauerRindone2013a}.
Recall that a path in an undirected graph is reduced if it does not use twice consecutively the same edge.

\begin{proposition}
\label{proposition6.6}
Let $S$ be an acyclic set.
For any bifix code  $X\subset S$, the following assertions hold.
\begin{enumerate}
\item[\rm (i)] The incidence graph $\G_X$ is acyclic.
\item[\rm (ii)] The intersection of $P' = P\setminus\{\varepsilon\}$ (resp. $Q'=Q\setminus\{\varepsilon\}$) with each connected component of $\G_X$ is a suffix (resp. prefix) code.
\item[\rm (iii)] For every reduced path $(v_1,u_1,\ldots,u_n,v_{n+1})$ in $\G_X$ with $u_1,\ldots,u_n\in P'$ and $v_1,\ldots,v_{n+1}\in Q'$, the longest common prefix of $v_1,v_{n+1}$ is a proper prefix of all $v_1,\ldots,v_n,v_{n+1}$.
\item[\rm(iv)] Symmetrically, for every reduced path $(u_1,v_1,\ldots,v_n,u_{n+1})$ in $\G_X$ with $u_1,\ldots,$ $u_{n+1}\in P'$ and $v_1,\ldots,v_n\in Q'$, the longest common suffix of $u_1,u_{n+1}$ is a proper suffix of $u_1,u_2,\ldots,u_{n+1}$.
\end{enumerate}
\end{proposition}

\subsection{Cosets}
Let $X$ be a symmetric set.
We use the incidence graph to define an equivalence relation $\gamma_X$ on the set $P$ of proper prefixes of $X$, called the \emph{coset equivalence} of $X$, as follows. 
It is the relation defined by $p\equiv q\bmod \gamma_X$ if there is a path (of even length) from $1\otimes p$ to $1\otimes q$ or a path (of odd length) from $1\otimes p$ to $q^{-1}\otimes 1$ in the incidence graph $\G_X$.
It is easy to verify that, since $X$ is symmetric, $\gamma_X$
is indeed an equivalence. The class of the empty word $\varepsilon$
is reduced to $\varepsilon$.
  This definition is an extension
to symmetric sets of the equivalence denoted $\theta_X$
introduced in~\cite{BerstelDeFelicePerrinReutenauerRindone2012}.

The following statement is the generalization to symmetric bifix codes
of Proposition 6.3.5 in~\cite{BerstelDeFelicePerrinReutenauerRindone2012}.
We denote by $\langle X\rangle$ the subgroup generated by $X$.
\begin{proposition}\label{propositionTheta}
Let $X$ be a symmetric bifix code and let $P$ be the set of its proper
prefixes. Let $\gamma_X$ be the coset equivalence of $X$ and 
let $H=\langle X\rangle$.
For any $p,q\in P$, if $p\equiv q\bmod\gamma_X$, then $Hp=Hq$.
\end{proposition}
\begin{proof}
Assume that there is a path of even length from $p$ to $q$.
If the path has length $2$, then we have $pr,qr\in X$ for some suffix $r$
of $X$. This implies $pq^{-1}\in H$ and thus $Hp=Hq$. The general case
follows by induction. In the case where there is a path of odd
length from $p$ to $q^{-1}$, there is a path of even length from
$p$ to $r$ and an edge from $r$ to $q^{-1}$ for some $r\in P$. Then $Hp=Hr$ by the 
preceding argument. Since $rq^{-1}\in X$, we have $Hr=Hq$ and the conclusion
follows.
\end{proof}

We now use the coset equivalence $\gamma_X$ to define the \emph{coset automaton} $\C_X$ of a symmetric bifix code $X$ as follows. 
The vertices of $\C_X$ are the equivalence classes of $\gamma_X$. We denote by $\hat{p}$ the class of $p$.
There is an edge labeled $a\in A$ from $s$ to $t$ if for some $p\in s$ and $q\in t$ (that is, $s=\hat{p}$ and $t=\hat{q}$), one of the following cases occurs (see Figure~\ref{figureCoset}):

\begin{enumerate}
\item[(i)] $pa\in P$ and $pa\equiv q\bmod\gamma_X$,
\item[(ii)] or $pa\in X$ and $q=\varepsilon$.
\end{enumerate}

\begin{figure}[hbt]
\gasset{Nadjust=wh}
\centering
\begin{picture}(60,10)
\put(0,0){
\begin{picture}(25,10)
\node(1)(0,5){$\hat{\varepsilon}$}
\node(p)(15,5){$\hat{p}$}
\node(pa)(25,5){$\hat{pa}$}
\drawedge(1,p){$p$}
\drawedge(p,pa){$a$}
\node[Nframe=n](i)(12,0){(i)}
\end{picture}
}

\put(40,0){
\begin{picture}(20,10)
\node(p)(0,5){$\hat{p}$}
\node(1)(10,5){$\hat{\varepsilon}$}
\drawedge(p,1){$a$}
\node[Nframe=n](iii)(5,0){(ii)}
\end{picture}
}

\end{picture}
\caption{The edges of the coset automaton.}\label{figureCoset}
\end{figure}
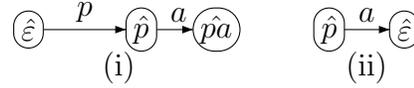

\begin{proposition}\label{propositionEdge}
Let $X$ be a symmetric bifix code, let $P$ be its set of proper prefixes and let $H = \langle X \rangle$.
If for $p,q \in P$ and a word $w \in A^*$ there is a path labeled $w$ from the class $\hat{p}$ to the class $\hat{q}$, then $Hpw = Hq$.
\end{proposition}
\begin{proof}
Assume first that $w$ is a letter $a\in A$.
It is easy to verify using Proposition~\ref{propositionTheta} that in the two
 cases of the definition
of an edge $(\hat{p},a,\hat{q})$, one has $Hpa=Hq$. Since the coset does
not depend on the representative in the class, this implies the conclusion.
The general case follows easily by induction.
\end{proof}
Let $A$ be an alphabet with an involution $\theta$.
A directed graph with edges labeled in $A$ is called \emph{symmetric} if there
is an edge from $p$ to $q$ labeled $a$ if and only if there is
an edge from $q$ to $p$ labeled $a^{-1}$.

If $\G$ is a symmetric graph and $v$ is a vertex of $\G$, the set of
reductions of the labels of paths from $v$ to $v$ is a subgroup of
$G_\theta$ called the subgroup \emph{described} by $\G$ with respect
to $v$.

A symmetric graph  is called  \emph{reversible}
if for every pair  of edges of the form $(v,a,w),(v,a,w')$,
 one has $w=w'$ (and the symmetric implication since the graph is symmetric).

\begin{proposition}\label{propositionCoset}
Let $S$ be a specular set and let $X\subset S$ be a finite symmetric bifix code.
The coset automaton $\C_X$ is reversible.
Moreover the subgroup described by $\C_X$ with respect to the class of the empty word is the group generated by $X$.
\end{proposition}
\begin{proof}
It is easy to verify that  the words of $X$ are labels of paths from $\hat{\varepsilon}$ to $\hat{\varepsilon}$ which do not pass by $\hat{\varepsilon}$ in between.
Thus the group described by $\C_X$
with respect to $\hat{\varepsilon}$ contains $H=\langle X\rangle$. 

By Proposition~\ref{propositionEdge}, if there is a path from the class of $p$
to the class of $q$ labeled $w$, then $Hpw=Hq$. Thus if $w$ belongs to the group
described by $\C_X$ (w.r.t. $\hat{\varepsilon}$), it is in $H$. We have thus proved that the coset automaton describes
$H$. 

Let us show now that $\C_X$ is reversible.
First, it is symmetric since $X$ is symmetric. Let
us show that if $(v,a,w)$ and $(v,a,w')$ are edges of $\C_X$, then 
$w=w'$.
Consider $p,p'\in P$ such that $p\equiv p'\bmod \gamma_X$. Assume that
there is an edge labeled $a$ from $\hat{p}=\hat{p'}$ to $\hat{q}$
and  to  $\hat{q'}$. 
\paragraph{Case 1} Suppose that $pa,p'a\in P$. We have to show that
$pa\equiv p'a\bmod \gamma_X$. Let $u,v$ be such that
$pau,p'av\in X$. It is not possible that
there exists a path of odd length from $p$ to $p'^{-1}$ in the incidence graph $\G_X$.
Indeed, assume that $p\in S_{i,j}$ and $a\in S_{j,k}$. Let $(p,u_1,\ldots,u_{2m},p'^{-1})$ with $m\ge 0$ be a path of odd length from $p$ to $p'^{-1}$.
Then each $u_{2t}$ for $1\le t\le m$ is in $S_{i_t,j}$ and each $u_{2t+1}$
for $0\le t\le m-1$ is in $S_{j,\ell_t}$ for some $i_t,\ell_t\in\{0,1\}$.
Then $p'^{-1}\in S_{j,\ell_m}$ and thus
$p'\in S_{1-\ell_m,1-j}$. But then we cannot have $p'a\in S$.
Thus there is a path of even length from $p$ to $p'$ in $\G_X$.
This implies that there is a path of even length of the
form $(au,p,\ldots,p',av)$. 
Thus by Proposition~\ref{proposition6.6} (iii), there is a path
of even length from $pa$ to $p'a$. 
This implies that $pa\equiv p'a\bmod\gamma_X$.

\paragraph{Case 2}
Assume now that $pa\in P$ and $p'a\in X$.
For the same reason as in Case 1, there cannot exist a path of odd length from $p$ to $p'$.
Thus there is a path of even length from $p$ to $p'$.
By Proposition~\ref{proposition6.6} (iii), this is not possible since otherwise we would have for some word $u$, a path $(au,p,\ldots,p',a)$ and $a$ is not a proper prefix of the last term of the sequence.

The case where $pa\in X$ and $p'a\in P$ is symmetrical.
Finally, if $pa,p'a\in X$, we have $q=q'=\varepsilon$.

This shows that if $(v,a,w)$ and $(v,a,w')$ are edges of $\C_X$, then
$w=w'$. Since $\C_X$ is symmetric, it follows that 
if $(v,a,w)$ and $(v',a,w)$ are edges of $\C_X$, then $v=v'$. Thus
$\C_X$ is reversible.
\end{proof}
\begin{example}
Let $T$ be the linear involution of Example~\ref{exampleInvolution3}
and let $S=\LL(T)$. Let $X$ be the set of words of length $3$ of
$S$ (see Figure~\ref{figureSetS}), which is a symmetric bifix code.
The incidence graph $\G_X$ is represented in Figure~\ref{figureIncidence}.
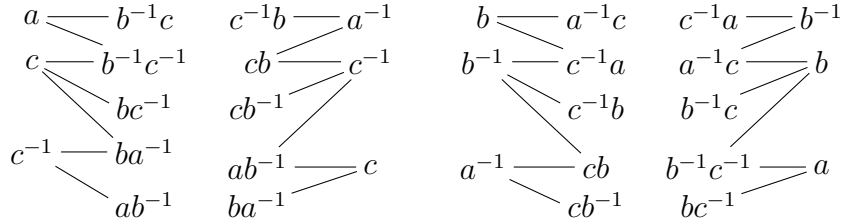
\begin{figure}[hbt]
\centering
\gasset{Nadjust=wh,AHnb=0,Nframe=n}
\begin{picture}(120,35)
\put(0,0){
\begin{picture}(20,30)
\node(a)(0,30){$a$}\node(bbarc)(15,30){$b^{-1}c$}
\node(c)(0,24){$c$}\node(bbarcbar)(15,24){$b^{-1}c^{-1}$}
\node(bcbar)(15,18){$bc^{-1}$}
\node(cbar)(0,12){$c^{-1}$}\node(babar)(15,12){$ba^{-1}$}
\node(abbar)(15,5){$ab^{-1}$}

\drawedge(a,bbarc){}\drawedge(a,bbarcbar){}
\drawedge(c,bbarcbar){}\drawedge[eyo=-2](c,bcbar){}\drawedge[eyo=-2](c,babar){}
\drawedge(cbar,babar){}\drawedge[eyo=-2](cbar,abbar){}
\end{picture}
}
\put(30,0){
\begin{picture}(20,30)
\node(abar)(15,30){$a^{-1}$}\node(cbarb)(0,30){$c^{-1}b$}
\node(cbar)(15,24){$c^{-1}$}\node(cb)(0,24){$cb$}
\node(cbbar)(0,18){$cb^{-1}$}
\node(c)(15,10){$c$}\node(abbar)(0,10){$ab^{-1}$}
\node(babar)(0,5){$ba^{-1}$}

\drawedge(cbarb,abar){}\drawedge(cb,abar){}
\drawedge(cb,cbar){}\drawedge(cbbar,cbar){}\drawedge(abbar,cbar){}
\drawedge(abbar,c){}\drawedge(babar,c){}
\end{picture}
}
\put(60,0){
\begin{picture}(20,30)
\node(b)(0,30){$b$}\node(abarc)(15,30){$a^{-1}c$}
\node(bbar)(0,24){$b^{-1}$}\node(cbara)(15,24){$c^{-1}a$}
\node(cbarb)(15,18){$c^{-1}b$}
\node(abar)(0,10){$a^{-1}$}\node(cb)(15,10){$cb$}
\node(cbbar)(15,5){$cb^{-1}$}

\drawedge(b,abarc){}\drawedge(b,cbara){}
\drawedge(bbar,cbara){}\drawedge[eyo=-2](bbar,cbarb){}\drawedge(bbar,cb){}
\drawedge(abar,cb){}\drawedge[eyo=-2](abar,cbbar){}
\end{picture}
}
\put(90,0){
\begin{picture}(20,30)
\node(bbar)(15,30){$b^{-1}$}\node(cbara)(0,30){$c^{-1}a$}
\node(b)(15,24){$b$}\node(abarc)(0,24){$a^{-1}c$}
\node(bbarc)(0,18){$b^{-1}c$}
\node(a)(15,10){$a$}\node(bbarcbar)(0,10){$b^{-1}c^{-1}$}
\node(bcbar)(0,5){$bc^{-1}$}

\drawedge(cbara,bbar){}\drawedge(abarc,bbar){}
\drawedge(abarc,b){}\drawedge(bbarc,b){}\drawedge(bbarcbar,b){}
\drawedge(bbarcbar,a){}\drawedge(bcbar,a){}
\end{picture}
}
\end{picture}
\caption{The incidence graph of $X$.}\label{figureIncidence}
\end{figure}
The coset automaton $\C_X$ is represented in Figure~\ref{figureCosetAutomaton}
(we only represent one of the edges labeled $a$ and $a^{-1}$, the other one
is understood).
The vertex $2$ is the class corresponding to the first two trees
in Figure~\ref{figureIncidence}. The vertex $3$ corresponds to the two
last ones.
\begin{figure}[hbt]
\centering
\gasset{Nadjust=wh}
\begin{picture}(30,35)
\node(1)(0,15){$1$}\node(2)(30,30){$2$}\node(3)(25,0){$3$}

\drawedge[curvedepth=3](1,2){$a,c$}\drawedge[curvedepth=3](2,1){$c$}
\drawedge[curvedepth=3](1,3){$b$}\drawedge[curvedepth=3](3,1){$a,b$}
\drawedge[curvedepth=3](2,3){$a$}\drawloop[loopangle=0](2){$b$}
\drawloop[loopangle=0](3){$c$}
\end{picture}
\caption{The coset automaton.}\label{figureCosetAutomaton}
\end{figure}
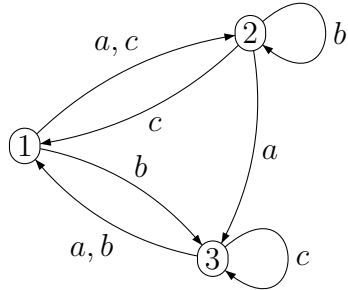
\end{example}

\subsection{Saturation Theorem}
Let $H$ be a subgroup of the specular group $G_\theta$ and let $S$
be a specular set on $A$ relative to $\theta$. The set of \emph{prime}
words in $S$ with respect to $H$
 is the set of nonempty  words in $H\cap S$ without a proper 
nonempty prefix in $H\cap S$. Note that the set of prime words with
respect to $H$
 is a symmetric  bifix code. One may verify that
it is
actually the unique bifix code $X$ such that $X\subset S\cap H\subset X^*$.

The following statement is a generalization of Theorem 5.2 in~\cite{BertheDeFeliceDolceLeroyPerrinReutenauerRindone2013a} (Saturation Theorem).

\begin{theorem}[Saturation Theorem]
\label{theoremSaturation}
Let $S$ be an acyclic laminary set.
Any finite symmetric bifix code $X\subset S$ is the set of prime words in $S$ with respect to the subgroup $\langle X\rangle$.
Moreover $\langle X\rangle\cap S=X^*\cap S$.
\end{theorem}
\begin{proof}
Let $H=\langle X\rangle$ and let $Y\subset S$ be the set of prime words with
respect to $H$.
Then $Y$ is a symmetric bifix code and thus it is free by Theorem~\ref{theoremFreeness}.
Since, by Proposition~\ref{propositionCoset}, the coset automaton $\C_X$ is reversible, any reduced word
is the label of at most one reduced path in $\C_X$. Since any word
of $X$ is the label of a  reduced  path from $\hat{\varepsilon}$ to $\hat{\varepsilon}$ in $\C_X$ which does not pass by $\hat{\varepsilon}$ inbetween, this implies that $X\subset Y$.
But any $y\in Y$ is the reduction of some product
$x_1x_2\cdots x_n$ with $x_i\in X$. Since $Y$ is free and contains $X$, this
implies $n=1$ and $y\in X$. Thus $X=Y$.

The last assertion follows from the fact that, since
$X$ is the set of prime words in $S$ with respect to $H$, one has $H\cap S\subset X^*$.
\end{proof}

Note  that the hypothesis that $X$ is symmetric is necessary, as shown in the following example.
\begin{example}
Let $A=\{a,b,a^{-1},b^{-1}\}$. 
Let $S$ be the set of factors of  $(ab^{-1})^\omega\cup (a^{-1}b)^\omega$
(we denote as usual by $x^\omega$ the infinite word $xxx\cdots$).
Then $S$ is an acyclic laminary set. The set $X=\{a,ba^{-1}\}$ is a  bifix
code but it is not the set of prime words with respect to
$\langle X\rangle$ since $b\in \langle X\rangle\cap S$.
\end{example}

\section{Bifix codes and monoidal bases}
\label{sectionFiniteIndex}
In this section we prove the Finite Index Basis Theorem (Theorem~\ref{theoremFiniteIndex}) and a converse (Theorem~\ref{propositionConverseFIB}).

\subsection{Finite Index Basis Theorem}
The following result is the counterpart for specular sets of the result holding for recurrent tree sets of characteristic $1$ (see~\cite[Theorem 4.4]{BertheDeFeliceDolceLeroyPerrinReutenauerRindone2013b}).
The proof is very similar to that of Theorem 4.4 in~\cite{BertheDeFeliceDolceLeroyPerrinReutenauerRindone2013b} and we omit some details.

\begin{theorem}[Finite Index Basis Theorem]
\label{theoremFiniteIndex}
Let $S$ be a recurrent specular set and let $X\subset S$ be a finite symmetric bifix code.
Then $X$ is an $S$-maximal bifix code of $S$-degree $d$ if and only if it is a monoidal basis of a subgroup of index $d$.
\end{theorem}
The following result is a complement to Theorem 4.4.3 in~\cite{BerstelDeFelicePerrinReutenauerRindone2012}, asserting that if $S$
is a recurrent set, any finite bifix code $X\subset S$ is contained in a finite $S$-maximal bifix code $Z$.
It shows that when $X$ is symmetric, then $Z$ can be chosen symmetric.

\begin{theorem}
\label{theoremSymCompletion}
Let $S$ be a recurrent laminary set.
Any finite symmetric bifix code $X\subset S$ is contained in a finite symmetric $S$-maximal bifix code.
\end{theorem}
\begin{proof}
Let $X\subset S$ be a finite symmetric bifix code which is not $S$-maximal.
Since $X$ is finite, the number $d=\max\{d_X(w)\mid w\in X\}$ is finite.
By Theorem 4.3.12 of~\cite{BerstelDeFelicePerrinReutenauerRindone2012}, $X$ is the kernel of some $S$-maximal bifix code $Z$ of $S$-degree $d+1$.
Since $S$ is recurrent, by Theorem 4.4.3 of~\cite{BerstelDeFelicePerrinReutenauerRindone2012}, $Z$ is finite.
Let us show that $Z$ is symmetric. Indeed, we have by Theorem 4.3.11 in~\cite{BerstelDeFelicePerrinReutenauerRindone2012}, $d_Z(w)=\min\{d+1,d_X(w)\}$. Since $X$ is symmetric, we have $d_X(w)=d_X(w^{-1})$ for any $w\in S$.
Indeed, $(q,x,p)$ is a parse of $w$ if and only if $(p^{-1},x^{-1},q^{-1})$ is a parse of $w^{-1}$.
Thus $d_Z(w)=d_Z(w^{-1})$. This implies that $Z$ is symmetric.
\end{proof}

\begin{proofof}{of Theorem~\ref{theoremFiniteIndex}}
Assume first that $X$ is a finite symmetric  $S$-maximal bifix code of $S$-degree  $d$.
Let $P$ be the set of proper prefixes of $X$.
Let $H$ be  the subgroup generated by $X$.

Let $u\in S$ be a word such that $d_X(u)=d$, or, equivalently, which is not an internal factor of $X$.
Since $u$ can be replaced by any of its right extensions, we may assume that $u$ is odd.
Let $Q$ be the set formed of the $d$ suffixes of $u$ which are in $P$.

Let us first show that the cosets $Hq$ for $q\in Q$ are
disjoint. Indeed, $Hp\cap Hq\ne\emptyset$ implies $Hp=Hq$.
Any $p,q\in Q$ are comparable for the suffix order. 
Assuming that $q$ is longer than $p$, we have
$q=tp$ for some $t\in P$. Then $Hp=Hq$ implies $Ht=H$ and thus $t\in
H\cap S$. By Theorem~\ref{theoremSaturation}, since $S$ is acyclic
and $X$ is symmetric,
 this implies $t\in X^*$
and thus $t=\varepsilon$. Thus $p=q$.

Let
\begin{displaymath}
  V=\{v\in G_\theta\mid Qv\subset HQ\}
\end{displaymath}
where the products $Qv$ and $HQ$ are understood in the  group $G_\theta$
(that is, with
reduction).

For any $v\in V$ the map $p\mapsto q$ from $Q$ into itself
defined by $pv\in Hq$
is a permutation of $Q$. Indeed, suppose that for
$p,q\in Q$, one has $pv,qv\in Hr$ for some $r\in Q$. Then $rv^{-1}$ is
in $Hp\cap Hq$ and thus $p=q$ by the above argument.

The set $V$ is a subgroup of $G_\theta$. Indeed, $1\in V$. Next, let $v\in V$. Then
for any $q\in Q$, since $v$ defines a permutation of $Q$, there
is a $p\in Q$ such that $pv\in
Hq$. Then $q v^{-1} \in Hp$. This shows that $v^{-1}\in V$.
Next, if $v,w\in V$, then $Qvw\subset HQw\subset HQ$ and thus $vw\in
V$.

We show that the set $\RR_S(u)$ is contained in $V$. 
Let $y\in\RR_S(u)$.
Since $uy$ ends with $u$, and since $u$ is not an internal
factor of $X$, for any $p\in Q$, we have
$py=xq$ for some $x\in X^*$ and $q\in Q$. Therefore $y\in V$.

By Theorem~\ref{theoremReturns}, the group generated by $\RR_S(u)$ is the even
subgroup.  Thus $V$ contains the even subgroup. But $V$ contains odd
words. Indeed, let $v\in S$ be such that $uvu^{-1}\in S$. Then $v$
is odd by Proposition~\ref{propositionOddReturns}. Moreover, for any
$p\in Q$ there is some $q\in Q$ such that $pvq^{-1}\in X^*$. This implies
that $pv\in X^*q$ and thus $v$ is in $V$. Since the even subgroup is of index $2$, it is maximal in $G_\theta$ and we conclude that $V=G_\theta$.

Thus $Qw\subset HQ$ for any $w\in G_\theta$. 
Since $\varepsilon \in Q$, we have in particular $w\in HQ$ for any $w\in G_\theta$.
Thus $G_\theta=HQ$. Since $\Card(Q)=d$, and since the right cosets $Hq$ for
$q\in Q$ are pairwise disjoint, this shows that $H$ is a subgroup of
index $d$.
By Theorem~\ref{corollaryCardinality}, we have $\Card(X)-2=
d(\Card(A)-2)$. But since $X$ generates $H$,   and since $X$ contains
the inverses of its elements, this implies by Proposition~\ref{propSymBasis}
that $X$ is a monoidal basis of $H$.

Assume conversely that the finite bifix code $X\subset F$ is a monoidal
basis of the
group $H=\langle X\rangle$ and that $\langle X\rangle$ has index $d$. Since $X$ is a monoidal
basis, by Schreier's Formula, we have $\Card(X)= (k-2)d+2$, where
$k=\Card(A)$. The case $k=1$ is straightforward; thus we assume
$k\ge2$. 
By
Theorem~\ref{theoremSymCompletion},  there is a finite symmetric $S$-maximal bifix
code $Y$ containing $X$. Let $e$ be the $S$-degree of $Y$. By the
first part of the proof, $Y$ is a monoidal basis of a subgroup $K$ of index $e$
of $G_\theta$.  In particular, it has $(k-2)e+2$ elements. Since
$X\subset Y$, we have $(k-2)d+2\le (k-2)e+2$ and thus $d\le e$. On the
other hand, since $H$ is included in $K$, $d$ is a multiple of $e$ and
thus $e\le d$. We conclude that $d=e$ and thus that $X=Y$.
\end{proofof}

Note that when $X$ is not symmetric, the index of the subgroup generated
by $X$ may be  different of $d_S(X)$, as shown in the following example.
\begin{example}
Let $T$ be as in  Example~\ref{exampleInvolution3}
and let $S=\LL(T)$. The set $X=\{a,ba^{-1},bc^{-1},b^{-1}c,b^{-1}c^{-1},
a^{-1}c,cb,cb^{-1},c^{-1}ab^{-1},c^{-1}b\}$ is an $S$-maximal bifix
code of $S$-degree $2$.
Since $b,c\in\langle X\rangle$, the group generated by $X$ is the free group on $A$.
\end{example}

The following consequence of Theorem~\ref{theoremFiniteIndex} is the
counterpart for specular sets of Theorem 5.10 in~\cite{BertheDeFeliceDolceLeroyPerrinReutenauerRindone2013m}. We give in~\cite{BertheDelecroixDolcePerrinReutenauerRindone2014} a geometric proof
and interpretation of Theorem~\ref{theoremGroupCode} for the natural
coding of a linear involution.

\begin{theorem}\label{theoremGroupCode}
Let $S$ be a recurrent specular set.
For any subgroup $H$ of finite index of the group $G_\theta$, the set of 
prime
words in $S$ with respect to $H$  is a monoidal basis of $H$.
\end{theorem}
\begin{proof}
Let $X$ be the set of prime words in $S$ with respect to $H$. The set $X$
is a symmetric bifix code 
and the number of parses of a word of $S$ is at most equal to the
index $d$ of $H$ in $G_\theta$. Indeed, let $(v,x,u)$ and $(v',x',u')$ be
two parses of a word $w\in S$. If $v,v'$ are in the same left coset of $H$,
then the two interpretations are equal. Indeed, assume that $|v|\ge|v'|$
and set $v=v's$. Then $s\in H$ and thus $s\in X^*$, which implies $s=1$
by definition of a parse. Therefore $X$ is an $S$-maximal bifix
code by~\cite[Theorem 4.2.8]{BerstelDeFelicePerrinReutenauerRindone2012}.

By Theorem~\ref{theoremFiniteIndex}, $X$ is a monoidal basis of
a subgroup $K$ of index $e$. Since $K\subset H$, the index of $K$
is a multiple of the index of $H$. Since $e\le d$, we conclude that
$e=d$ and that $K=H$.
\end{proof}
We illustrate Theorem~\ref{theoremGroupCode} with the following 
interesting example.
\begin{example}
Let $T$ be as in Example~\ref{exampleInvolution3}
and let $S=\LL(T)$. Let $G$ be the group of even words in $F_A$.
It is a subgroup of index $2$. The set of prime words in $S$ 
with respect to $G$ 
is the set $Y=X\cup X^{-1}$ with
\begin{displaymath}
X=\{a,ba^{-1}c,bc^{-1},b^{-1}c^{-1},b^{-1}c\}.
\end{displaymath}
Actually, the transformation induced by $T$ on the set $I\times \{0\}$
(the upper part of $\hat{I}$ in Figure~\ref{figureLinear3}) is the
interval exchange transformation represented in Figure~\ref{figureRetourHaut}.
Its upper intervals are the $I_x$ for $x\in X$.
\begin{figure}[hbt]
\centering
\gasset{AHnb=0,Nadjust=wh}
\begin{picture}(100,20)
\node(h0)(0,10){}\node(babarc)(23.6,10){}\node(bcbar)(47.2,10){}
\node(bbarcbar)(61.8,10){}\node(bbarc)(85.4,10){}\node(h1)(100,10){}
\node(b0)(0,0){}\node(cb)(14.1,0){}\node(cbarb)(38.2,0){}\node(cbarabbar)(52.6,0){}
\node(abar)(76.4,0){}\node(b1)(100,0){}

\drawedge[linecolor=red,linewidth=1](h0,babarc){$a$}
\drawedge[linecolor=blue,linewidth=1](babarc,bcbar){$ba^{-1}c$}
\drawedge[linecolor=forestgreen,linewidth=1](bcbar,bbarcbar){$bc^{-1}$}
\drawedge[linecolor=golden,linewidth=1](bbarcbar,bbarc){$b^{-1}c^{-1}$}
\drawedge[linecolor=gray,linewidth=1](bbarc,h1){$b^{-1}c$}
\drawedge[linecolor=light-gray,linewidth=1](b0,cb){$c^{-1}b$}
\drawedge[linecolor=yellow,linewidth=1](cb,cbarb){$cb$}
\drawedge[linecolor=green,linewidth=1](cbarb,cbarabbar){$cb^{-1}$}
\drawedge[linecolor=cyan,linewidth=1](cbarabbar,abar){$c^{-1}ab^{-1}$}
\drawedge[linecolor=magenta,linewidth=1](abar,b1){$a^{-1}$}
\end{picture}
\caption{The transformation induced on the upper level.}\label{figureRetourHaut}
\end{figure}
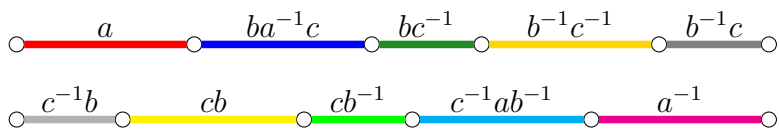
This corresponds to the fact that the words of $X$ correspond to
the first returns to $I\times\{0\}$ while the words of $X^{-1}$ correspond
to the first returns to $I\times\{1\}$.
\end{example}

\subsection{A converse of the Finite Index Basis Theorem}
The following is a converse of Theorem~\ref{theoremFiniteIndex}.

\begin{theorem}
\label{propositionConverseFIB}
Let $S$ be a recurrent laminary set of factor complexity $p_n=n(\Card(A)-2)+2$.
If $S\cap A^n$ is a monoidal basis of the subgroup $\langle A^n\rangle$ for all $n\ge 1$, then $S$ is a specular set.
\end{theorem}
\begin{proof}
Consider $w\in S$ and set $m=|w|$.
The set $X=(AwA \cup Aw^{-1}A) \cap S$ is closed by taking inverses and it is included in $Y=S\cap A^{m+2}$.
Since $Y$ is a monoidal basis of a subgroup, $X\subset Y$ is a monoidal basis of the subgroup  $\langle X\rangle$.

This implies that the graph $\E(w)$ is acyclic.
Indeed, assume that the path $(a_1,b_1,\ldots, a_p,b_p,a_1)$ is a cycle in $\E(w)$ with $p\ge 2$, $a_i\in L(w)$, $b_i\in R(w)$ for $1\le i\le p$ and $a_1\ne a_p$.
Then $a_1wb_1,a_2wb_1,\ldots,$ $a_pwb_p,a_1wb_p\in X$.
But
\begin{displaymath}
a_1wb_1(a_2wb_1)^{-1}a_2wb_2\cdots a_pwb_p(a_1wb_p)^{-1} = \varepsilon,
\end{displaymath}
with $a_jwb_j(a_{j+1}wb_j)^{-1} = a_ja^{-1}_{j+1} \not = \varepsilon$ (otherwise $a_j = a_{j+1}$), contradicting the fact that $X$ is a monoidal basis.

Since $p_n=n(\Card(A)-2)+2$, we have $s_n=\Card(A)-2$ and $b_n=0$ for all $n>0$.
By Proposition~\ref{propCANT}, it implies that $m(w)=0$ for all nonempty words $w$.
Since $\E(w)$ is acyclic, we conclude that $\E(w)$ is a tree.

Finally, since $\E(\varepsilon)$ is acyclic, and since $m(\varepsilon)=-1$, the graph $\E(\varepsilon)$ has two connected components which are trees.
\end{proof}

\section*{References}
\addcontentsline{toc}{section}{References}

\bibliographystyle{plain}
\bibliography{specular}

\end{document}